\documentclass[12pt]{article}
\usepackage{amsfonts,array}
\usepackage{multicol,subcaption,makecell}
\usepackage[font=footnotesize,labelfont=bf]{caption}
\usepackage{afterpage}
\usepackage{amsmath,bbm,dsfont,mathrsfs,mathtools,appendix}
\usepackage{amssymb}
\usepackage{graphicx}
\usepackage{fullpage}%
\usepackage{enumitem}
\usepackage{physics}
\usepackage{framed}
\usepackage[table]{xcolor}
\usepackage{authblk}

\usepackage[colorlinks=true,linkcolor=blue,citecolor=green,plainpages=false,pdfpagelabels]%
{hyperref}%
\hypersetup{
	colorlinks,
	linkcolor={blue!60!green},
	citecolor={green!50!yellow!75!black},
	urlcolor={blue!80!black},
	linktoc=all
}
\usepackage{cleveref} 
\providecommand{\U}[1]{\protect\rule{.1in}{.1in}}

\newtheorem{theorem}{Theorem}

\newtheorem{corollary}[theorem]{Corollary}

\newtheorem{definition}[theorem]{Definition}

\newtheorem{lemma}[theorem]{Lemma}

\newtheorem{proposition}[theorem]{Proposition}

\newenvironment{proof}[1][Proof]{\noindent\textbf{#1.} }{\ \rule{0.5em}{0.5em}}

\newcommand{\im}{\operatorname{im}}

\newcommand{\M}[1]{\mathbb{M}_{#1}(\mathbb{C})}
\newcommand{\B}[1]{\mathcal{L}({#1})}
\newcommand{\iden}{\mathbbm{1}}

\usepackage{array}
\usepackage{amsmath,amssymb,mathtools,bm}
\usepackage{booktabs}
\usepackage{graphicx}

\usepackage{tikz}
\usetikzlibrary{arrows.meta,positioning,fit,calc}

\tikzset{
  qblock/.style={
    draw,
    rounded corners=2pt,
    minimum height=7mm,
    minimum width=16mm,
    align=center,
    font=\scriptsize,
    fill=black!3
  },
  qsmall/.style={
    draw,
    rounded corners=2pt,
    minimum height=6mm,
    minimum width=13mm,
    align=center,
    font=\scriptsize,
    fill=black!3
  },
  qarrow/.style={-{Latex[length=1.8mm]}, thick},
  qgroup/.style={
    draw,
    dashed,
    rounded corners=3pt,
    inner sep=3pt
  }
}

\newcommand{\id}{{\rm{id}}} 

\newcommand{\R}{\mathbbm{R}}
\newcommand{\C}{\mathbb{C}}

\newcommand{\N}{\mathbb{N}}
\newcommand{\cM}{\mathcal{M}}

\newcommand{\cT}{\mathcal{T}}

\newcommand{\cE}{\mathcal{E}}

\newcommand{\cN}{\mathcal{N}}

\newcommand{\cD}{\mathcal{D}}

\newcommand{\E}{\mathbb{E}}

\def\>{{\rangle}}
\def\<{{\langle}}
\newcommand{\be}{\begin{equation}}
	\newcommand{\ee}{\end{equation}}
\newcommand{\bea}{\begin{eqnarray}}
	\newcommand{\eea}{\end{eqnarray}}

\newcommand{\eps}{\varepsilon}

\newcommand{\comment}[1]{}

\newcommand{\supp}{\operatorname{supp}}
\newcommand{\CPTP}{{\rm{CPTP}}}

\newcommand{\Bsa}[1]{\mathcal{L}^{\operatorname{sa}}(#1)}

\newcommand{\ID}{\operatorname{ID}}

\numberwithin{equation}{section}
\numberwithin{theorem}{section}

\usepackage{color}
\definecolor{colorthree}{rgb}{0.01,0.51,0.93}

\usepackage{setspace}

\usepackage{newunicodechar}
\newunicodechar{ρ}{\rho}
\newunicodechar{β}{\beta}
\newunicodechar{σ}{\sigma}
\newunicodechar{λ}{\lambda}
\newunicodechar{Λ}{\Lambda}
\newunicodechar{μ}{\mu}
\newunicodechar{ν}{\nu}
\newunicodechar{ψ}{\psi}
\newunicodechar{ϕ}{\varphi}
\newunicodechar{Φ}{\cN}
\newunicodechar{φ}{\cN}
\newunicodechar{π}{\pi}
\newunicodechar{α}{\alpha}
\newunicodechar{ϵ}{\epsilon}
\newunicodechar{ε}{\varepsilon}
\newunicodechar{δ}{\delta}
\newunicodechar{ω}{\omega}
\newunicodechar{Δ}{\Delta}
\newunicodechar{Σ}{\Sigma}

\title{The entanglement-assisted transmission capacity is a strong converse bound for identification}
\author[1]{Satvik Singh}

\affil[1]{Department of Mathematics, Technical University of Munich, Garching, Germany}

\affil[1]{Munich Center for Quantum Science and Technology (MCQST), Munich, Germany}

\date{}

\begin{document}

\maketitle

\begin{abstract}
Classical identification via a noisy channel is a communication task in which the receiver is not required to reconstruct the full transmitted message, but only to decide whether it coincides with a message of interest. This relaxation allows the number of identifiable messages to grow doubly exponentially with the blocklength. For quantum channels, the resulting (doubly exponential) identification capacity $C_{\ID}$ can strictly exceed the ordinary (exponential) transmission capacity $C$. 

In this paper, we prove that the entanglement-assisted transmission capacity $C_E$ is a strong converse bound for this task: $C_{\ID}\leq C_E$. For sufficiently low-noise channels, this bound can also be achieved via the Hayden-Winter (quantum) identification + fingerprinting codes. This yields an exact characterization $C_{\ID}=C_E$ of identification capacity for such channels. However, for general channels, we prove that this upper bound can be strict. We exhibit an explicit family of transpose-depolarizing channels for which $C_{\ID}<C_E$. As a consequence, we also obtain the first example of strict superadditivity of the identification
capacity $C_{\ID}$.

\end{abstract}

\tableofcontents

\section{Introduction}

The operational meaning of information in a communication setting depends strongly on the question the receiver is asked to answer. In Shannon's original model of communication, the
receiver is required to reconstruct the message sent by the sender
\cite{Shannon1948, Cover2005book}, thus answering the question ``What was the transmitted message?" For a memoryless noisy channel used \(n\)
times, the maximum number of reliably transmissible messages grows exponentially
with \(n\), and the optimal exponential growth rate is the transmission capacity $C$. For classical channels, this is given by the maximum input-output mutual information of the channel \cite{Shannon1948}.
For quantum channels, the corresponding transmission capacity is given
by the regularized Holevo information \cite{Schumacher1997HSW,Holevo1998HSW,Wilde2016}.

Identification, introduced by Ahlswede and Dueck, asks for a different kind of communication \cite{Ahlswede1989ID}. Instead of reconstructing the message, the receiver
is only required to answer a binary question of the form: ``Was the transmitted
message equal to \(j\)?'' This relaxation has a dramatic consequence: the number
of identifiable messages can grow doubly exponentially with the blocklength $n$ \cite{Ahlswede1989ID}. The
corresponding optimal double-exponential growth rate is the identification capacity $C_{\ID}$. For
classical memoryless channels, the identification
capacity is equal to Shannon's transmission capacity $C=C_{\ID}$ \cite{Ahlswede1989ID}. Identification
over quantum channels was subsequently developed in \cite{Lober1999thesis-ID,Ahlswede2002strong-ID,
Winter2005identification-hybrid-memory,Winter2013survey-ID}.

Quantum channels make the identification problem richer than its classical
counterpart. For a classical channel, all tests in an identification code are automatically
compatible because the channel output is classical. For a quantum channel,
however, an identification decoder consists of independent binary tests
\(\{D_j,\iden-D_j\}\), one for each possible message \(j\). These tests need not
arise as coarse-grainings of a single POVM. Requiring such a common
refinement leads to the \emph{simultaneous} identification capacity
\(C_{\ID}^{\rm sim}\) \cite{Lober1999thesis-ID}. In general, Ahlswede and Dueck's construction \cite{Ahlswede1989ID} lifts any classical transmission code to an exponentially larger simultaneous identification code, which shows that $ C_{\ID}\geq C_{\ID}^{\rm sim}\geq C $.
Moreover, unrestricted quantum identification can exploit incompatible measurements, and this leads to phenomena with no classical analogue. For example, for the noiseless identity channel on
\(\C^d\),
\begin{equation}\label{eq:C(id)>Csim(id)}
2\log d=C_{\ID}(\id_d) >  C^{\rm sim}_{\ID} (\id_d) =   C(\id_d)= \log d,
\end{equation}
so the classical identification capacity of a quantum channel can be strictly larger
than its ordinary classical transmission capacity
\cite{Winter2005identification-hybrid-memory, Atif2024CIDstrongconverse, Colomer2025zero-entropy}.

A natural benchmark for this enhanced identification power is the
entanglement-assisted classical capacity \(C_E\). Operationally, this is the optimal (exponential) rate of reliable classical message transmission through a given channel when the sender and receiver have unlimited shared entanglement. Both \(C_{\ID}\) and \(C_E\) are equal for the noiseless channel, and more generally, for conditional expectations onto hybrid classical-quantum memories \cite{Winter2005identification-hybrid-memory}. The connection is further strengthened by the task of quantum identification. In this
task, for every
pure output state \(\varphi\), Bob is required to simulate the binary test $(\varphi,\iden-\varphi)$, i.e. answer the question ``Is the transmitted state equal to $\varphi$?'' on the noisy outputs with approximately the same statistics as in the noiseless case. Hayden and Winter showed that quantum identification is governed by a weak
decoupling principle: Bob can identify the input state whenever the environment
approximately forgets its identity \cite{Hayden2012QID-achievability}. Their
capacity theorem implies that, for sufficiently low-noise channels (see Definition~\ref{def:low-noise}), the
quantum identification capacity $Q_{\ID}$ equals \(C_E\). Combining this with quantum fingerprinting
\cite{Buhrman2001fingerprint} gives the following achievability bound
\begin{equation}\label{eq:CID>=QID>=C_E-intro} 
C_{\ID}(\cN)\geq Q_{\ID}(\cN)= C_E(\cN)    
\end{equation}
for such low-noise channels $\cN$ (see Section~\ref{sec:CID-QID-CE} for details).

This led Winter to raise the possibility that
\begin{equation}\label{eq:CID>=CE-intro}
   C_{\ID}(\cN)\ge C_E(\cN)
\end{equation}
might hold for all quantum channels \cite{Winter2013survey-ID}. The question
has remained open partly because of the fact that obtaining converse bounds for
identification has proved to be difficult. Prior to the current work, no matching upper bound of the form $C_{\ID}(\cN)\leq C_E(\cN)$ was known. General converse bounds are available, for instance via quantum soft covering \cite{Atif2024CIDstrongconverse}, but they typically contain dimension-dependent terms and are not tight in very noisy regimes. Recent work has developed sharper geometric converse bounds \cite{singh2026gaussianmeanwidthstrong} (see also \cite{ye2026strongconverseboundsclassical}). Nevertheless, an exact general formula for \(C_{\ID}\) remains out of reach, and all existing examples support \eqref{eq:CID>=CE-intro}.

\subsection{Our contribution}
In this paper, we obtain two main results.
\begin{itemize}
    \item First, we establish the universal strong converse bound (c.f. Theorem~\ref{theorem:CID-leq-CE})
    \begin{equation}\label{eq:CID<=CE-intro}
    C_{\ID}(\cN)\leq C_E(\cN)    
    \end{equation}
    for all quantum channels $\cN$.
    \item Second, we construct explicit channels $\cN$ for which this converse bound is strict: $C_{\ID}(\cN) < C_E(\cN)$, thus proving that \eqref{eq:CID>=CE-intro} \emph{cannot} hold universally (c.f. Proposition~\ref{prop:main}). As a consequence, we also obtain, to the best of our knowledge, the first example of non-additivity of $C_{\ID}$ (c.f. Corollary~\ref{cor:CID-nonadditive}).
\end{itemize}

Combining the universal converse \eqref{eq:CID<=CE-intro} with the Hayden-Winter achievability \eqref{eq:CID>=QID>=C_E-intro} gives a complete characterization of identification capacity for low-noise channels (c.f. Corollary~\ref{cor:low-noise}):
\begin{equation}
    C_{\ID}(\cN)
    =Q_{\ID}(\cN)
    =C_E(\cN).
\end{equation}

\subsection{Proof techniques}

The proof of the universal converse $C_{\ID}\leq C_E$ builds on the Euclidean geometric method based on the notion of Gaussian mean width introduced recently in
\cite{singh2026gaussianmeanwidthstrong}. The new technical step is to relate this geometric converse bound to the channel max information \cite{Fang2020channel-max-information} (see Lemma~\ref{lemma:mu-star-Imax}), followed by appropriate smoothing and use of the asymptotic equipartition property (AEP) of the channel smooth max-information
\cite[Theorem 8]{Fang2020channel-max-information} (see Theorem~\ref{theorem:CID-leq-CE}). 

Next, to obtain the strict separation $C_{\ID} < C_E$, we consider the depolarizing and transpose-depolarizing channels $\cD_p, \cD_q^{\mathsf T}:\M{d}\to \M{d}$:
\begin{align}
    \cD_p(X)
    &=
    (1-p)X+p\Tr(X)\frac{\iden}{d} \\
    \cD_q^{\mathsf T}(X)
    &=
    (1-q)X^{\mathsf T}+q\Tr(X)\frac{\iden}{d}.
\end{align}
A simple yet crucial observation is that since the two channels are related via transposition:
\begin{equation}
    (\cD_p^{\mathsf T})^{\otimes n} = {\mathsf T}^{\otimes n} \circ \cD_p^{\otimes n},
\end{equation}
which maps the set of quantum states bijectively onto itself, their identification capacities must coincide (see Proposition~\ref{prop:main}). A direct application of the converse $C_{\ID}\leq C_E$ then yields the desired separation at $p=q=\frac{d}{d+1}$:
   \begin{align}
        C_{\ID}\left(\cD_{\frac{d}{d+1}}^{\mathsf T}\right) = C_{\ID}\left(\cD_{\frac{d}{d+1}}\right)  &\leq C_{E}\left(\cD_{\frac{d}{d+1}}\right) \\ 
        &=\log d-\frac{d-1}{d}\log(d+1) \\
        & < 1 + \log \frac{d}{d+1} = C_E \left(\cD_{\frac{d}{d+1}}^\top \right).
    \end{align}
Note that 
\begin{equation}
    C_{\ID}\left(\cD_{\frac{d}{d+1}}^{\mathsf T}\right) \longrightarrow 0 \quad \text{while} \quad C_E \left(\cD_{\frac{d}{d+1}}^\top \right)\longrightarrow 1 \quad \text{as} \quad d\to \infty.
\end{equation}
Intuitively, the separation between $C_{\ID}$ and $C_E$ for this transpose-depolarizing channel arises from the fact that while the channel
is almost completely noisy for large $d$ on bare inputs, it retains non-trivial correlations when applied on half of a maximally entangled state. This observation also motivates us to formulate the stronger `image-optimized' version of the $C_E$ converse, see Section~\ref{sec:discussion} for details.

As a consequence, we also obtain the first example of non-additivity of $C_{\ID}$. Tensoring the transpose-depolarizing channel with a sufficiently large noiseless channel makes the product channel low-noise in the sense of the Hayden-Winter quantum identification theorem \eqref{eq:CID>=QID>=C_E-intro}. The product channel therefore achieves its entanglement-assisted rate for identification via \eqref{eq:CID>=QID>=C_E-intro}, while the transpose-depolarizing factor alone satisfies $C_{\ID} <C_E$. This yields strict superadditivity of $C_{\ID}$ (see Corollary~\ref{cor:CID-nonadditive}).

\subsection{Outline}

The paper is organized as follows. 
\begin{itemize}
    \item In Section~\ref{sec:prelim}, we collect the preliminary material required for the rest of the paper.
\item In Section~\ref{sec:main}, we state and prove our main results. 
\item  In Section~\ref{sec:discussion}, we conclude with a discussion.
\end{itemize}

\section{Preliminaries} \label{sec:prelim}

\hspace{3pt} Consider a scenario in which two parties, say Alice (sender) and Bob (receiver), want to use a noisy quantum channel $\cN:\B{A}\to \B{B}$ to communicate information between them. Here, $A\cong \C^{d_A}, B\cong \C^{d_B}$ are finite-dimensional complex Hilbert spaces, $\B{A}$ denotes the space of all linear operators acting
on $A$, and $\cN$ is a linear, completely positive, and trace-preserving map. The real space of Hermitian operators in $\B{A}$ is denoted by $\Bsa{A}$. The Hilbert-Schmidt inner product on $\B{A}$ is denoted by
\begin{equation}\label{eq:HS}
    \langle X,Y\rangle := \Tr (X^{\dagger}Y), \qquad X,Y\in \B{A},
\end{equation}
and the Hilbert-Schmidt adjoint of $\cN:\B{A}\to \B{B}$ is defined by the relation
\begin{equation}
    \langle Y,\cN(X)\rangle
    =
    \langle\cN^{*}(Y),X\rangle, \qquad X\in \B{A}, Y\in \B{B}.
\end{equation}
The trace-norm, Schatten $\alpha$-norm for $\alpha\in(1,\infty)$, and the operator
norm on $\B{A}$ are denoted by $\norm{\cdot}_1$, $\norm{\cdot}_\alpha$, and
$\norm{\cdot}_\infty$, respectively \cite{Bhatia1997matrix}. We will use $\cD(A)$ to denote the set of quantum states acting on $A$, i.e. positive semi-definite operators in $\Bsa{A}$ with unit trace. The set of full-rank states in $\cD(A)$ is denoted by $\cD_+(A)$. We denote the cardinality of a set $S$ by $\abs{S}$. For $n\in\N$, we denote $[n]:=\{0,1,2,\ldots,n-1\}$. 

\subsection{Transmission vs identification}

In the task of classical information \emph{transmission} from Alice to Bob, Alice first encodes a message $i\in \{0,1,\ldots ,N-1\}=:[N]$ from a list of possible messages in a quantum state $\rho_i\in \cD(A^{\otimes n})$, sends it via $\cN^{\otimes n}$ to Bob, who performs a measurement via a POVM $\{\Lambda_i\}_{i\in [N]}\subseteq \B{B^{\otimes n}}$, $\Lambda_i\geq 0$, $\sum_i \Lambda_i = \iden_{B^{\otimes n}}$ in an attempt to decode exactly which message Alice intended to send (see Figure~\ref{fig:transmission}). It is known from the original work of Shannon \cite{Shannon1948, Cover2005book, Wilde2016} that the maximum number of messages that can be reliably transmitted via $n$ independent uses of $\cN$ scales \emph{exponentially} with $n$. This motivates the definition of transmission capacity.

\begin{figure}[ht]
\centering

\resizebox{\linewidth}{!} {%
\begin{tikzpicture}[node distance=4mm and 8mm]

\node (msg) {$i\in[N]$};
\node[qblock, right=of msg] (enc) {Encoder};
\node[right=of enc] (rho) {$\rho_i$};
\node[qblock, right=of rho] (chan) {Noisy channel\\ $\mathcal N^{\otimes n}$};
\node[right=of chan] (outstate) {$\mathcal N^{\otimes n} (\rho_i)$};
\node[qblock, right=of outstate] (dec) {Decoder POVM\\$\{\Lambda_i\}_{i=1}^N$};
\node[right=of dec] (what) {$\hat{i}\in[N]$};

\draw[qarrow] (msg) -- (enc);
\draw[qarrow] (enc) -- (rho);
\draw[qarrow] (rho) -- (chan);
\draw[qarrow] (chan) -- (outstate);
\draw[qarrow] (outstate) -- (dec);
\draw[qarrow] (dec) -- (what);

\node[qgroup, fit=(msg)(enc)(rho), label={[font=\scriptsize]above:sender}] {};
\node[qgroup, fit=(outstate)(dec)(what), label={[font=\scriptsize]above:receiver}] {};

\node[below=3mm of dec, font=\scriptsize]
{\(\Pr[\hat{i}= i] \geq 1-\lambda \) large};

\node[below=3mm of enc, font=\scriptsize]{Code parameters $(n,N,\lambda)$};

\end{tikzpicture}%
}

\caption{Schematic of classical message transmission over $n$ uses of a noisy channel $\cN$. A message $i\in [N]$ is encoded into a quantum state $\rho_i$, transmitted through $n$ uses of the channel, and decoded by a POVM to produce an estimate $\hat{i}$. }
\label{fig:transmission}
\end{figure}
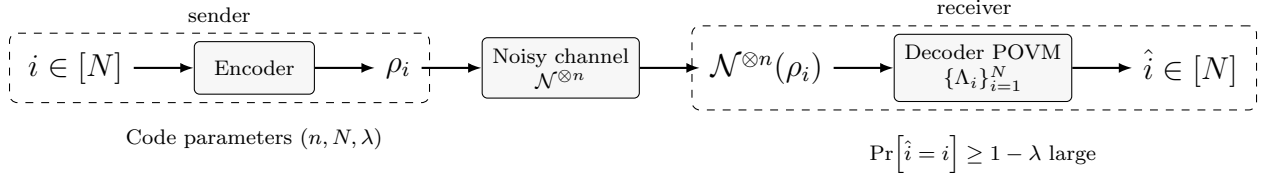

\begin{definition}\label{def:Transmission-codes}
    A $(n,N,\lambda)$ (classical) \emph{transmission code} for a channel $\cN: \B{A}\to \B{B}$ is defined by encoding quantum states $\{\rho_i\}_{i\in [N]} \subseteq \cD (A^{\otimes n})$ and a decoding measurement $\{\Lambda_i\}_{i\in [N]} \subseteq \B{B^{\otimes n}}$, $\Lambda_i\geq 0$, $\sum_i \Lambda_i = \iden_{B^{\otimes n}}$, such that
    \begin{align}\label{eq:lambda}
        \forall i \in [N] : \qquad \Tr ( \cN^{\otimes n}(\rho_i) \Lambda_i ) &\geq 1-\lambda .
    \end{align}
    For a given $n\in \N, \lambda\in [0,1)$, the maximum size of all $(n,N,\lambda)$ codes for $\cN$ is 
    \begin{equation}
        N_{(n, \lambda)}(\cN) := \max \{N : \exists (n,N, \lambda) \text{ classical transmission code for } \cN \},
    \end{equation}
    and the \emph{classical transmission capacity} is defined as 
    \begin{align}
        C(\cN) := \inf_{0<\lambda <1} \liminf_{n\to \infty} \frac{1}{n} \log N_{(n, \lambda)}(\cN) .
    \end{align}
    If the sender and receiver also have access to unlimited pre-shared entanglement, the corresponding maximum number of classical
messages that can be transmitted via $n$ uses of $\cN$ with maximal decoding error at most
$\lambda$ is denoted by
$N^E_{(n,\lambda)}(\cN)$. The associated \emph{entanglement-assisted classical transmission
capacity} is defined as (see \cite{Wilde2016, watrous2018theory} for precise definitions):
\begin{equation}\label{eq:CE-operational-definition}
    C_E(\cN)
    :=
    \inf_{0<\lambda<1}
    \liminf_{n\to\infty}
    \frac1n
    \log N^E_{(n,\lambda)}(\cN).
\end{equation}
\end{definition}

In the task of classical message \emph{identification}, instead of decoding the full message, Bob is only interested in deciding/identifying whether the sent message $i$ is equal to a fixed message $j$ or not (see Figure~\ref{fig:identification}). This small change leads to a drastic difference in how the code size scales with the block-length. Indeed, the maximum number of messages that can be reliably identified via $n$ independent uses of $\cN$ scales \emph{doubly exponentially} with $n$ \cite{Ahlswede1989ID, Lober1999thesis-ID, Ahlswede2002strong-ID}. In order to see this, we start with an ordinary transmission code with $N$ codewords for $n$ uses of the channel, where $N\approx 2^{nR}$ grows exponentially in $n$. We then choose a large family of subsets $S_1,\ldots,S_{N_{\rm ID}}\subseteq [N]$, each of the same size, but with pairwise intersections much smaller than their size. To send identification message $j$, Alice randomly selects one transmission codeword from $S_j$. The decoder for query $j$ accepts if the ordinary transmission decoder outputs an index belonging to $S_j$. The small overlap condition ensures that the probability of falsely accepting $j$ when another message was sent is small, while reliable transmission ensures that $j$ is accepted with high probability when it was sent. Since the number of such almost-disjoint subsets can be shown to be exponential in $N$, the resulting number of identification messages is doubly exponential in $n$: $N_{
\rm ID} \approx 2^{N} \approx 2^{2^{nR}}$. We refer the reader to \cite{Ahlswede1989ID, Colomer2025zero-entropy} for more details. This motivates the definition of identification capacity.

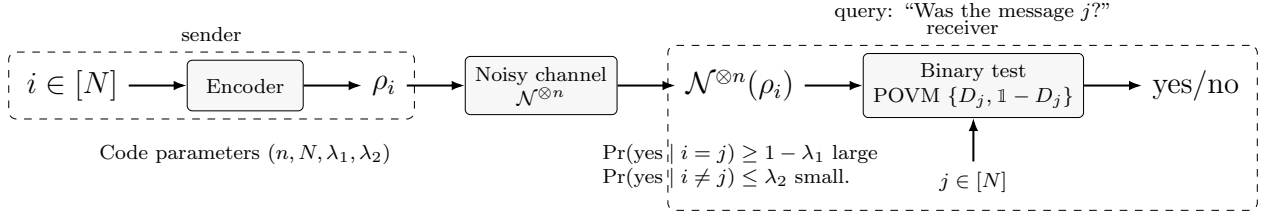
\begin{figure}[ht]
\centering
\resizebox{\linewidth}{!}{%
\begin{tikzpicture}[node distance=2mm and 8mm]

\node (msg) {$i\in[N]$};
\node[qblock, right=of msg] (enc) {Encoder};
\node[right=of enc] (rho) {$\rho_i$};
\node[qblock, right=of rho] (chan) {Noisy channel\\ $\mathcal N^{\otimes n}$};
\node[right=of chan] (outstate) {$\mathcal N^{\otimes n}(\rho_i)$};
\node[qblock, right=of outstate] (test) {Binary test\\POVM \(\{D_j,\iden -D_j\}\)};
\node[right=of test] (ans) {yes/no};

\draw[qarrow] (msg) -- (enc);
\draw[qarrow] (enc) -- (rho);
\draw[qarrow] (rho) -- (chan);
\draw[qarrow] (chan) -- (outstate);
\draw[qarrow] (outstate) -- (test);
\draw[qarrow] (test) -- (ans);

\node[above=3mm of test, font=\scriptsize] (query)
  {query: ``Was the message \(j\)?''};

\node[below=6mm of test, font=\scriptsize] (jinput) {$j\in[N]$};
\draw[qarrow] (jinput) -- (test);

\node[qgroup, fit=(msg)(enc)(rho), label={[font=\scriptsize]above:sender}] {};
\node[qgroup, fit=(outstate)(test)(ans)(jinput), label={[font=\scriptsize]above:receiver}] {};

\node[below=3mm of outstate, align=left, font=\scriptsize]
{\(\Pr(\mathrm{yes}\mid i=j) \geq 1-\lambda_1\) large \\
 \(\Pr(\mathrm{yes}\mid i\neq j)\leq \lambda_2\) small.};

\node[below=3mm of enc, font=\scriptsize]{Code parameters $(n,N,\lambda_1, \lambda_2)$};

\end{tikzpicture}%
}

\caption{Schematic of classical message identification over \(n\) uses of a noisy quantum channel $\cN$. A message \(i\in[N]\) is encoded into a quantum state \(\rho_i\), transmitted via $n$ uses of the channel, and the receiver performs the binary test \(\{D_j,\iden-D_j\}\) to decide whether the transmitted message was \(j\in [N]\).}
\label{fig:identification}
\end{figure}

\begin{definition}\label{def:IDcodes}
    A $(n,N,\lambda_1,\lambda_2)$ (classical) \emph{identification (ID) code} for a quantum channel $\cN: \B{A}\to \B{B}$ is defined by pairs $\{(\rho_i, D_i) \}_{i\in [N]}$ of encoding states $\rho_i \in \cD (A^{\otimes n})$ and decoding effects $D_i\in \B{B^{\otimes n}}$, $0_{B^{\otimes n}} \leq D_i \leq \iden_{B^{\otimes n}}$, such that
    \begin{align}
        \forall i \in [N] : \qquad \Tr ( \cN^{\otimes n}(\rho_i) D_i ) &\geq 1-\lambda_1, \label{eq:lambda_1} \\ 
        \forall j\neq i\in [N] : \qquad \Tr (\cN^{\otimes n}(\rho_i) D_j) &\leq \lambda_2  \label{eq:lambda_2}.
    \end{align}
    For a given $n\in \N, \lambda_1,\lambda_2\in [0,1)$, the maximum size of all $(n,N,\lambda_1, \lambda_2)$ codes for $\cN$ is\footnote{To avoid trivialities, we will always assume that $\lambda_1+\lambda_2<1$.} 
    \begin{equation}
        N_{(n, \lambda_1, \lambda_2)}(\cN) := \max \{N : \exists (n,N, \lambda_1, \lambda_2) \text{ classical ID code for } \cN \},
    \end{equation}
    and the classical \emph{ID capacity} is defined as
    \begin{align}
        C_{\ID}(\cN) &:= \inf_{\substack{\lambda_1, \lambda_2 >0 \\ \lambda_1+\lambda_2<1}} \liminf_{n\to \infty} \frac{1}{n} \log \log N_{(n, \lambda_1, \lambda_2)}(\cN) .
    \end{align}
\end{definition}

If the binary tests
\(\{D_i\}_{i\in [N]}\) in the above definition arise as coarse-grainings of a single common POVM, then the corresponding \emph{simultaneous} identification capacity is denoted by
\(C_{\ID}^{\rm sim}(\cN)\) \cite{Lober1999thesis-ID}. From the previous discussion, any transmission code can be lifted to an exponentially large simultaneous identification code \cite{Ahlswede1989ID}. Hence, for all channels $\cN$,
\begin{equation}\label{eq:CID>=C}
    C_{\ID}(\cN) \geq C_{\ID}^{\rm sim}(\cN) \geq C(\cN).
\end{equation}

\subsection{Quantum identification} \label{sec:CID-QID-CE}

For low-noise channels, the achievability bound from Eq.~\eqref{eq:CID>=C} can be improved. In order to show this, we first recall the notion of quantum identification \cite{Winter2005identification-hybrid-memory, Hayden2012QID-achievability}. Let $K\cong \C^{d_K}$ and denote by $\mathbb P(K)$ the set of rank-one projectors (or pure states) on $K$. A \emph{visible} $\varepsilon$-quantum identification code for $\cN^{\otimes n}$ with code space $K$ consists of an encoding map $\mathcal E:\mathbb P(K)\to \cD(A^{\otimes n})$ and, for every $\varphi\in\mathbb P(K)$, a binary test $0\leq D_\varphi\leq \iden_{B^{\otimes n}}$ such that, for all $\psi,\varphi\in\mathbb P(K)$,
\begin{equation}\label{eq:QID-eps}
\left| \Tr(\psi\varphi)
-
\Tr\left[\cN^{\otimes n}(\mathcal E(\psi))D_\varphi\right]
\right|
\leq \varepsilon .
\end{equation}
Thus, after receiving the channel output, Bob must be able to simulate the two-outcome test $(\varphi,\iden-\varphi)$, i.e. answer the question: ``Is the transmitted state equal to $\varphi$?'' with approximately the same statistics as in the noiseless case. If the encoding map extends to a completely positive trace-preserving map (quantum channel) $\mathcal E:\cD(K)\to\cD(A^{\otimes n})$, the code is called \emph{blind}. We denote by $M_{(n,\varepsilon)}(\cN)$ and $M^v_{(n,\varepsilon)}(\cN)$ the largest possible dimensions of $K$ for $\eps$-blind and $\eps$-visible quantum ID codes for $n$ uses of $\cN$, respectively. The corresponding capacities are defined as
\begin{align}
Q_{\ID}(\cN)
&:=
\inf_{\varepsilon>0}
\liminf_{n\to\infty}
\frac{1}{n}\log M_{(n,\varepsilon)}(\cN), \\
Q_{\ID,v}(\cN)
&:=
\inf_{\varepsilon>0}
\liminf_{n\to\infty}
\frac{1}{n}\log M^v_{(n,\varepsilon)}(\cN).
\end{align}

The connection with classical identification is obtained by concatenating a quantum ID code $K$ with a \emph{fingerprinting} construction \cite{Buhrman2001fingerprint}. Choose a large family of pure states $\{\psi_i \}_{i\in [N]}\subset\mathbb P(K)$ with small pairwise overlaps. To send the classical message $i$, Alice encodes the state $\psi_i$ using the quantum ID encoder $\cE:\mathbb P(K)\to \cD(A^{\otimes n})$. To test whether the message was $j$, Bob uses the quantum ID test $D_{\psi_j}$. If $i=j$, the noiseless acceptance probability is $\Tr(\psi_i\psi_i)=1$, so the quantum ID condition \eqref{eq:QID-eps} gives acceptance probability at least $1-\varepsilon$. If $i\neq j$, the noiseless acceptance probability $\Tr(\psi_i\psi_j)= \abs{\langle \psi_i | \psi_j\rangle}^2$ is small by the fingerprinting property, and the quantum ID condition \eqref{eq:QID-eps} adds only $\varepsilon$ error. Since one can choose exponentially many such fingerprints in the dimension of $K$ \cite{Buhrman2001fingerprint}, a quantum ID code with $d_K\approx 2^{nR}$ gives a classical ID code with $N\approx 2^{2^{nR}}$. Hence (see \cite{Hayden2012QID-achievability, Winter2013survey-ID} for details):
\begin{equation}\label{eq:CID>=QID}
C_{\ID}(\cN)\geq Q_{\ID,v}(\cN)\geq Q_{\ID}(\cN).
\end{equation}

We now recall the capacity formula for (blind) quantum identification. For an input pure state $|\phi\rangle_{RA}$, set $\omega_{RB}=(\id_R\otimes\cN_{A\to B})(\phi_{RA})$, and write $I(R:B)_\omega:= S(R)_\omega + S(B)_\omega - S(RB)_{\omega}$ for the \emph{mutual information} and $I(R\rangle B)_\omega=S(B)_\omega-S(RB)_\omega$ for the \emph{coherent information} \cite{Wilde2016}.  Using a random coding argument along with the weak decoupling principle \cite[Theorem 1]{Hayden2012QID-achievability}, Hayden and Winter showed that the quantum ID capacity is governed by the mutual information, but only over inputs with positive coherent information \cite{Hayden2012QID-achievability}:
\begin{equation}\label{eq:QID=I(R:B)}
Q_{\ID}(\cN)
=
\lim_{n\to\infty}
\frac{1}{n}
Q_{\ID}^{(1)}(\cN^{\otimes n}),
\qquad
Q_{\ID}^{(1)}(\cN)
=
\sup_{\phi: I(R\rangle B)_\omega>0}
I(R:B)_\omega, 
\end{equation}
where the $\sup$ is declared to be zero if there are no inputs $\phi$ with $I(R\rangle B)_{\omega}>0$. Note that the fully optimized mutual information is exactly the entanglement-assisted classical transmission capacity of the channel \cite{Bennett1999BSST, Wilde2016, watrous2018theory}:
\begin{equation}\label{eq:C_E=I}
     C_E(\cN) = \max_{\phi} I(R:B)_{\omega}.
\end{equation}

\begin{definition}[Low-noise channels] \cite[Remark 16]{Winter2013survey-ID} \label{def:low-noise}
     We say that a channel $\cN:\B{A}\to \B{B}$ has sufficiently \emph{low noise} if there exists an input $|\phi_\star\rangle_{RA}$ maximizing $I(R:B)_\omega$ in \eqref{eq:C_E=I} such that $I(R\rangle B)_{\omega_\star}>0$, where $\omega_\star := (\id_R \otimes \cN_{A\to B})(\phi_{RA})$.
\end{definition}

For low noise channels, the entanglement-assisted classical capacity is given by \cite{Bennett1999BSST, Wilde2016}
\begin{equation}
    C_E(\cN) = \max_{\phi} I(R:B)_{\omega} = I(R:B)_{\omega_\star}
\end{equation}
and since the corresponding coherent information $I(R\rangle B)_{\omega_\star}>0$, the state is also admissible in the formula \eqref{eq:QID=I(R:B)}. Combining this with the additivity of $C_E$ \cite{Wilde2016} yields
\begin{equation}
Q_{\ID}(\cN)
=
C_E(\cN).
\end{equation}
Combining this further with \eqref{eq:CID>=QID} shows that for any such low-noise channel,
\begin{equation}\label{eq:CID>=QID>=C_E}
C_{\ID}(\cN)
\geq
Q_{\ID,v}(\cN)
\geq
Q_{\ID}(\cN)
=
C_E(\cN),
\end{equation}
thus improving the lower bound \eqref{eq:CID>=C} for such channels.
Motivated by this, the question of whether the bound 
\begin{equation}\label{eq:CID>=CE}
    C_{\ID}(\cN)\geq C_E(\cN)
\end{equation}
holds universally for all channels was posed in \cite{Winter2013survey-ID}. We quote: ``In fact, the few cases for which $C_{\ID}$ is known are consistent with the idea
that it is always equal to the entanglement-assisted classical capacity of the
channel. One might speculate that $C_{\ID}(\cN)\geq Q_{\ID, v}(\cN)\geq C_E(\cN)$ be
true for all channels, seeing that for sufficiently low noise we can prove it, and
that it is true for the amortized classical ID-capacity" \cite{Winter2013survey-ID}.

\subsection{Gaussian mean width Euclidean converse}

With the two achievability bounds \eqref{eq:CID>=C} and \eqref{eq:CID>=QID} on identification in place, let us now turn to the converse bounds. In this section, we will briefly describe the geometric technique developed recently in \cite{singh2026gaussianmeanwidthstrong} to obtain strong converse bounds on identification. For the other quantum soft covering based converse, we refer the reader to \cite{Atif2024CIDstrongconverse}.

Let us begin by recalling the basic principle behind proving converse bounds for identification. Let $\{(\rho_i,D_i)\}_{i\in [N]}$ be an $(n,N,\lambda_1,\lambda_2)$ ID code for a channel $\cN:\B{A}\to \B{B}$. Set $\omega_i := \cN^{\otimes n}(\rho_i)$. By the defining properties of the code and the variational characterization of the trace norm, we get for $i\neq j$:
\begin{align}
    d_{\Tr}(\omega_i,\omega_j)
    :=
    \frac{1}{2}\norm{\omega_i-\omega_j}_1
    &=
    \max_{0\leq D \leq \iden_{B^{\otimes n}}} \Tr\bigl((\omega_i-\omega_j)D\bigr) \\
    &\geq
    \Tr\bigl((\omega_i-\omega_j)D_i\bigr) \\
    &\geq 1- \lambda_1-\lambda_2. \label{eq:CID-separation} 
\end{align}
Hence, the output states $\{\omega_i\}_{i\in [N]}$ of any $(n,N,\lambda_1,\lambda_2)$ ID code for $\cN$ form a large separated subset (i.e. packing) of the image set
\begin{equation}\label{eq:N-image}
I_n(\cN):=\cN^{\otimes n} (\cD(A^{\otimes n})) = \{\cN^{\otimes n}(\rho) : \rho \in \cD(A^{\otimes n}) \} \subseteq \Bsa{B^{\otimes n}}
\end{equation}
with respect to the trace distance. Let us formally define the covering and packing numbers.

\begin{definition}\label{def:packing-covering}
Let $(X,d)$ be a metric space, let $S\subseteq X$ be a bounded subset. The closed ball of radius $\eps>0$ around $x\in X$ is denoted by $B_\eps(x) := \{ y\in X : d(x,y)\leq \eps \}$.
\begin{itemize}
    \item A (finite) subset $P\subseteq S$ is called an $\eps$\emph{-packing} of $S$ if $\forall x\neq y\in \mathscr{P}, \,d(x,y)>\eps$. The \emph{$\eps$-packing number} of $S$ with respect to the metric $d$ is defined by
\begin{equation}
    \mathscr P_{\eps}(S;d) := \max \left\{ |\mathscr{P}| : \mathscr{P}\subseteq S \, \text{ is an } \eps\text{-packing} \right\}.
\end{equation} 
    \item A (finite) subset $C\subseteq S$ is called an $\eps$\emph{-covering} of $S$ if $S \subseteq \bigcup_{x\in \mathscr{C}} B_\eps(x)$. The \emph{$\eps$-covering number} of $S$ with respect to the metric $d$ is defined by
\begin{equation}
    \mathscr C_{\eps}(S;d) := \min \left\{ |\mathscr{C}| : \mathscr{C}\subseteq S \, \text{ is an } \eps\text{-covering} \right\}.
\end{equation} 
\end{itemize}
For any $S\subseteq X$, the following holds true \cite{Vershynin2018HDP, Avidan2015analysis}:
   \begin{equation}\label{eq:cover-pack-inequality}
        \mathscr P_{2\eps}(S;d) \leq \mathscr C_{\eps}(S;d) \leq \mathscr P_{\eps}(S;d).
    \end{equation}
\end{definition}

We can now write the necessary condition in \eqref{eq:CID-separation} as follows.

\begin{lemma}\label{lemma:CID-separation}
For a channel $\cN:\B{A}\to \B{B}$, $\lambda_1,\lambda_2>0$ with $\lambda_1+\lambda_2<1$, and every  $0<\zeta< 1-\lambda_1-\lambda_2$, the following holds true:
\begin{equation}
    N_{(n,\lambda_1,\lambda_2)}(\cN)
    \leq \mathscr{P}_{\zeta}(I_n (\cN) ; d_{\Tr}) \leq 
    \mathscr{C}_{\zeta/2}\bigl(I_n(\cN); d_{\Tr}).
\end{equation}
\end{lemma}

Thus, upper bounds on covering numbers of the image $I_n(\cN)$ in trace distance yield converse bounds for identification. In Euclidean spaces, Sudakov inequality is an important tool to bound packing/covering numbers of a bounded set in terms of its Gaussian mean width, which quantifies how wide the set looks in random directions on average (see e.g. \cite{Ledoux1991prob-banach, Avidan2015analysis, Vershynin2018HDP} for detailed expositions).

\begin{definition}\label{def:gaussian-width}
The \emph{Gaussian mean width} of a bounded set $S\subseteq \mathbb R^m$ is defined as
\begin{equation}
    w_G(S):=\mathbb E \sup_{x\in S}\langle g,x\rangle,
\end{equation}
where $g\sim N(0,1_m)$ is a standard Gaussian random vector\footnote{i.e. $g$ is a $\R^m$-valued random variable whose coordinates (in any orthonormal basis) are independent and identically distributed (i.i.d) according to the standard Gaussian distribution $N(0,1)$.} in $\mathbb R^m$.
\end{definition}

\begin{lemma}[Sudakov inequality in $\R^m$]  \label{lemma:sudakov} \cite{Sudakov1971} \cite[Chapter 3]{Ledoux1991prob-banach} \cite[Chapter 7]{Vershynin2018HDP} \\
There exists a universal constant $C>0$ such that for every bounded set
$S\subseteq \mathbb R^m$ and $\eps>0$,
\begin{equation}
\log \mathscr{C}_{\eps}(S;d)\leq C\,\frac{w_G(S)^2}{\eps^2},
\end{equation} 
where $d(x,y):= \sqrt{\sum_i (x_i-y_i)^2}$ is the standard Euclidean distance in $\R^m$.
\end{lemma}

Note that the trace distance is not Euclidean. However, it can
be compared with any Euclidean norm on the output space at an appropriate
domination cost. More precisely,
let $W$ be a \emph{Euclidean structure} on $\Bsa{B}$, i.e. a Hilbert--Schmidt self-adjoint and positive definite linear map $W:\Bsa{B}\to\Bsa{B}$. It
induces the following inner product and norm
\begin{equation}\label{eq:W-Euclidean}
    \langle Y,Z\rangle_W
    :=
    \Tr\bigl(YW(Z)\bigr),
    \qquad
    \norm{Y}_W
    :=
    \sqrt{\langle Y,Y\rangle_W},
\end{equation}
on $\Bsa{B}$. The corresponding \emph{trace norm domination cost} is defined as
\begin{equation}\label{eq:KW}
    K(W)
    :=
    \sup_{Y\neq0}
    \frac{\norm{Y}_1}{\norm{Y}_W}.
\end{equation}
Thus, $K(W)$ is the optimal constant such that $\norm{Y}_1
    \leq
    K(W)\norm{Y}_W$ holds for all $Y\in \Bsa{B}$. We note an alternative characterization of $K(W)$ below. 

\begin{lemma}\label{lemma:KW-dual}
For a Euclidean structure $W:\Bsa B\to\Bsa B$,
\begin{equation}
    K(W)^2
    =
    \sup_{\substack{H=H^\dagger\\ \norm{H}_\infty\leq1}}
    \Tr\!\left(H  W^{-1}(H)\right).
\end{equation}
\end{lemma}

\begin{proof}
By duality of trace and operator norms \cite[Chapter 4]{Bhatia1997matrix}, we can write
\begin{align}
    K(W) = \sup_{Y\neq0}
    \frac{\norm{Y}_1}{\norm{Y}_W}
    &= 
    \sup_{Y\neq0}
    \sup_{\substack{H=H^\dagger\\ \norm{H}_\infty\leq1}}
    \frac{\Tr(HY)}{\norm{Y}_W} \\
    &=
    \sup_{\substack{H=H^\dagger\\ \norm{H}_\infty\leq1}}
    \sup_{Y\neq0}
    \frac{\Tr(HY)}{\sqrt{\Tr(YW(Y))}}.
\end{align}
Using Cauchy-Schwarz inequality in the Hilbert-Schmidt space $(\Bsa{B}, \langle \cdot , \cdot \rangle)$ shows
\begin{equation}
    \Tr(HY)
    =
    \Tr\!\left(
        W^{-1/2}(H)\,W^{1/2}(Y)
    \right) \leq \sqrt{\Tr\bigl(HW^{-1}(H)\bigr)\,
    \Tr\bigl(YW(Y)\bigr) },
\end{equation}
where equality is attained for $Y=W^{-1}(H)$. Consequently,
\begin{equation}
    \sup_{Y\neq0}
    \frac{\Tr(HY)}{\sqrt{\Tr(YW(Y))}}
    =
    \sqrt{\Tr\!\left(HW^{-1}(H)\right)},
\end{equation}
and we obtain the desired expression:
\begin{equation}
    K(W)
    =
    \sup_{\substack{H=H^\dagger\\ \norm{H}_\infty\leq1}}
    \sqrt{\Tr\!\left(H W^{-1}(H)\right)}.
\end{equation}
\end{proof}

For a bounded set $S\subseteq\Bsa{B}$, define its \emph{Gaussian mean width} with
respect to this $W$-Euclidean structure as
\begin{equation}
    w_{G,W}(S)
    :=
    \E\sup_{Y\in S}
    \langle G_W,Y\rangle_W,
\end{equation}
where $G_W$ is a standard Gaussian vector in
$\bigl(\Bsa{B},\langle\cdot,\cdot\rangle_W\bigr)$. Let $\cN:\B{A}\to\B{B}$ be a quantum channel. Its $W$-adjoint
$\cN^{*,W}:\Bsa{B}\to\Bsa{A}$ is defined by the relation
\begin{equation}\label{eq:N-Wadjoint}
    \langle Y,\cN(X)\rangle_W
    =
    \langle\cN^{*,W}(Y),X\rangle, \qquad X\in \Bsa{A}, Y\in \Bsa{B},
\end{equation}
where the inner product $\langle \cdot , \cdot \rangle$ on $\Bsa{A}$ is the Hilbert-Schmidt inner product \eqref{eq:HS}. The associated $W$-\emph{weighted singular operator} is defined as \cite[Section 3.1]{singh2026gaussianmeanwidthstrong}
\begin{equation}\label{eq:Q_N,W}
    Q_{\cN,W}
    :=
    \E\left[
        \left(\cN^{*,W}(G_W)\right)^2
    \right]
    \in\Bsa{A}.
\end{equation}
Below, we note some simple properties of the singular operator.

\begin{lemma}[Properties of the weighted singular operator]
\label{lemma:Q-general-W-properties}
Let $\cN:\B{A}\to\B{B}$ be a quantum channel and let
$W:\Bsa B\to\Bsa B$ be a Euclidean structure. 
\begin{itemize}
    \item For any
Hilbert--Schmidt orthonormal basis $\{H_a\}_a$ of $\Bsa B$,
\begin{align}
    Q_{\cN,W}
    &=
    \sum_a
    \left[
        \cN^*\bigl(W^{1/2}(H_a)\bigr)
    \right]^2
    \label{eq:Q-W-basis-expansion}\\
    &=
    \sum_{a,b}
    \left\langle H_a,W(H_b)\right\rangle
    \cN^*(H_a)\cN^*(H_b).
    \label{eq:Q-W-covariance-expansion}
\end{align}
In particular, $Q_{\cN,W}\geq0$ and the expression is
independent of the chosen basis.
\item The map $W\mapsto Q_{\cN, W}$ extends uniquely to a linear positivity-preserving map $W\mapsto \Xi_{\cN}(W)$ defined on all Hilbert-Schmidt
self-adjoint linear maps $W: \Bsa B \to \Bsa B$.
\end{itemize}
\end{lemma}

\begin{proof}
Let $G=\sum_a g_aH_a$ be any standard Gaussian vector in the
Hilbert-Schmidt space $(\Bsa B , \langle \cdot , \cdot \rangle )$, where the $g_a$ are
independent standard real Gaussian variables. Then, $G_W:=W^{-1/2}(G)$
is a standard Gaussian vector in
$(\Bsa B,\langle\cdot,\cdot\rangle_W)$. Hence,
\begin{align}
 Q_{\cN,W} =  \E \left[ (\cN^{*,W}(G_W))^2 \right]
    &=
   \E \left[ \left(\cN^*\bigl(W^{1/2}(G)\bigr) \right)^2 \right] \\
    &=
    \E \left[ \left( \sum_a g_a\,
    \cN^*\bigl(W^{1/2}(H_a)\bigr) \right)^2 \right] \\
    &= \sum_a
        \left( \cN^*\bigl(W^{1/2}(H_a)\bigr) \right)^2 \\
    &= \sum_{a,b}
    \left\langle H_a,W(H_b)\right\rangle
    \cN^*(H_a)\cN^*(H_b).
\end{align}
where we used $\cN^{*,W}=\cN^* \circ W$ and 
$\mathbb E[g_ag_b]=\delta_{ab}$.

The above expansion is linear in $W$, and hence defines the
claimed linear extension $\Xi_{\cN}$. To prove positivity-preserving property, let
$S\geq0$ and choose a spectral decomposition on the real
Hilbert--Schmidt space $\Bsa B$,
\begin{equation}
    S
    =
    \sum_k\lambda_k
    \ket{F_k}\!\bra{F_k},
    \qquad
    \lambda_k\geq0,
\end{equation}
where $\{F_k\}_k$ is Hilbert--Schmidt orthonormal and
$\ket{F_k}\!\bra{F_k}$ denotes the rank-one superoperator
$Y\mapsto\langle F_k,Y\rangle F_k$. Then
\begin{equation}
    \Xi_{\cN}(S)
    =
    \sum_k\lambda_k
    \left[\cN^*(F_k)\right]^2
    \geq0,
\end{equation}
because each $\cN^*(F_k)$ is Hermitian.
\end{proof}

For a quantum channel $\cN:\B{A}\to \B{B}$ and an $n$-fold Euclidean structure $W:\Bsa{B^{\otimes n}}\to \Bsa{B^{\otimes n}}$, it is shown in \cite{singh2026gaussianmeanwidthstrong} that  the $n$-shot $W$-Gaussian width of the image $I_n(\cN)$ is bounded from above as 
\begin{equation}
    w_{G,W}(I_n(\cN)) \leq \sqrt{2n\log d_A \norm{Q_{\cN^{\otimes n},W}}_{\infty} }. 
\end{equation}
Combining this with the domination bound $\norm{\cdot}_1 \leq K(W) \norm{\cdot}_{W}$ and Sudakov inequality (Lemma~\ref{lemma:sudakov}) gives a non-trivial method to control the covering numbers of the channel images $I_n(\cN)$ in trace distance, thus yielding a non-trivial method to obtain converse bounds on identification via Lemma~\ref{lemma:CID-separation}, see \cite{singh2026gaussianmeanwidthstrong} for more details. Below, we state the main result of \cite[Theorem 5.2]{singh2026gaussianmeanwidthstrong} in a one-shot form.

\begin{theorem}[Euclidean Gaussian mean width converse \cite{singh2026gaussianmeanwidthstrong}] 
\label{theorem:fully-optimized-one-shot-converse} \hspace{2pt} \\
Let $\cN:\B{A}\to\B{B}$ be a quantum channel, $n\in \N$ and $\lambda_1,\lambda_2>0$ such that $
\lambda_1+\lambda_2<1$. Then, for every $0<\zeta< 1-\lambda_1-\lambda_2$, the following holds true:
\begin{equation}\label{eq:fully-optimized-packing}
    \log
    N_{(n,\lambda_1,\lambda_2)}(\cN)
    \leq
    \frac{2nC\log d_A}{\zeta^2}\,
    \mu_*(\cN^{\otimes n}),
\end{equation}
where $C>0$ is a universal constant, and 
\begin{equation}\label{eq:mu-star}
    \mu_*(\cN)
    :=
    \inf_W
    \left[
        K(W)^2
        \norm{Q_{\cN,W}}_\infty
    \right],
\end{equation}
where $W$ ranges over all admissible Euclidean structures on $\Bsa{B}$.
\end{theorem}

\section{Main results}\label{sec:main}

\subsection{Universal entanglement-assisted capacity converse}
\label{sec:universal-CE-converse}

We now combine the Euclidean converse from Theorem~\ref{theorem:fully-optimized-one-shot-converse} with the channel smooth max-information and its
asymptotic equipartition property
\cite{Fang2020channel-max-information} to obtain a universal strong converse
bound $C_{\ID}\leq C_E$ in terms of the entanglement-assisted classical capacity.

For a quantum channel $\cN:\B{A}\to\B{B}$, define the \emph{unnormalized Choi operator}
\begin{equation}\label{eq:choi-def}
    J_{\cN}
    :=
    \sum_{a,b\in[d_A]}
    \ket{a}\!\bra{b}\otimes
    \cN\bigl(\ket{a}\!\bra{b}\bigr)
    \in\Bsa{A\otimes B},
\end{equation}
so that
$\Tr_B J_{\cN}=\iden_A$. Define the \emph{channel max-information} as \cite{Fang2020channel-max-information}:
\begin{equation}\label{eq:channel-Imax}
    I_{\max}(\cN)
    :=
    \log\inf\left\{
        \Tr V_B:
        V_B\geq0,
        J_{\cN}\leq \iden_A\otimes V_B
    \right\},
\end{equation}
and, for $\eps\geq0$, define its diamond-norm smoothed version by
\begin{equation}\label{eq:smooth-channel-Imax}
    I_{\max}^{\eps}(\cN)
    :=
    \inf_{\substack{
        \widetilde{\cN}:\B{A}\to\B{B}\;\CPTP\\[1mm]
        \frac12\norm{\widetilde{\cN}-\cN}_{\diamond}\leq\eps
    }}
    I_{\max}(\widetilde{\cN}).
\end{equation}
Recall that the \emph{diamond norm} of a linear map $\cN: \B{A}\to \B{B}$ is defined as
\begin{equation}
    \norm{\cN}_{\diamond} := \sup_{\norm{X}_1\leq 1} \norm{ (\id_R \otimes \cN_{A\to B})(X_{RA})}_1,
\end{equation}
where the supremum is over all $X\in \B{R\otimes A}$ and $d_R\in \mathbb{N}$ with $\norm{X}_1\leq 1$\cite[Chapter 3]{watrous2018theory}.

The channel max-information can be viewed as a one-shot version of the channel mutual information \eqref{eq:C_E=I}, and it
characterizes the one-shot cost of simulating a quantum channel using a noiseless quantum channel and no-signalling assistance \cite[Theorem 6]{Fang2020channel-max-information}. Crucially,
the smoothed max-information satisfies the following asymptotic equipartition property.

\begin{theorem}\cite[Theorem~8]{Fang2020channel-max-information}
    For every quantum channel $\cN:\B{A}\to \B{B}$,
    \begin{equation}\label{eq:channel-Imax-AEP}
            \lim_{\eps\to0}
    \lim_{n\to\infty}
    \frac1n
    I_{\max}^{\eps}(\cN^{\otimes n})
    =
    C_E(\cN).
    \end{equation}
\end{theorem}

The next lemma is the primary technical ingredient of our work, where we relate the optimized Euclidean quantity $\mu_*$ from Theorem~\ref{theorem:fully-optimized-one-shot-converse} to the
channel max-information $I_{\max}$ via a specifically chosen Euclidean structure \eqref{eq:W-bures}.

\begin{lemma}[Max-information converse]
\label{lemma:mu-star-Imax} \hspace{2pt} \\ For every quantum channel $\cN:\B{A}\to\B{B}$, the following holds true:
\begin{equation}\label{eq:mu-star-Imax}
    \mu_*(\cN)
    \leq
    2^{I_{\max}(\cN)+1}.
\end{equation}
\end{lemma}

\begin{proof}
Fix an arbitrary feasible operator $V_B\geq0$ in the optimization \eqref{eq:channel-Imax} defining $I_{\max}$, i.e.
$J_{\cN}\leq\iden_A\otimes V_B$. For $\delta>0$, the operator
$V_B+\delta\iden_B$ is full rank and remains feasible, while
$\Tr(V_B+\delta\iden_B)\to\Tr V_B$ as $\delta\downarrow0$. It therefore
suffices to prove
\begin{equation}\label{eq:mu-star-c}
    \mu_*(\cN)\leq2\Tr V_B
\end{equation}
for every full-rank $V_B\geq 0$ satisfying $J_{\cN}\leq\iden_A\otimes V_B$. 

Fix such a $V_B$, and set $c:=\Tr V_B$ and
$\sigma:=V_B/c\in\cD_+(B)$. Then
$J_{\cN}\leq D:=\iden_A\otimes V_B
=c\,\iden_A\otimes\sigma_B$. Consequently,
$R:=D^{-1/2}J_{\cN}D^{-1/2}$ is a positive contraction. Hence
$R^2\leq R$, and conjugating this inequality by $D^{1/2}$ gives
$J_{\cN}D^{-1}J_{\cN}\leq J_{\cN}$. Thus, the feasability condition
\begin{align}
 J_{\cN}\leq c \iden_A \otimes \sigma_B &\implies    J_{\cN}
    \bigl(\iden_A\otimes\sigma_B^{-1}\bigr)
    J_{\cN}
    \leq
    c\,J_{\cN} \\ 
    &\implies \Tr_B[J_{\cN}
    \bigl(\iden_A\otimes\sigma_B^{-1}\bigr)
    J_{\cN}] \leq c \iden_A. \label{eq:Jsigma J<=J}
\end{align}
  
Now, consider the Euclidean structure $W_{\sigma, \mathsf{B}}:\Bsa{B}\to \Bsa{B}$:
\begin{equation}\label{eq:W-bures}
    W_{\sigma, \mathsf{B}}
    :=
    2(L_\sigma+R_\sigma)^{-1},
    \qquad
    L_\sigma(Y):=\sigma Y,
    \quad
    R_\sigma(Y):=Y\sigma.
\end{equation}
Since $\sigma$ is full rank, $W_{\sigma,\mathsf{B}}$ is a Hilbert--Schmidt self-adjoint,
positive definite linear map on $\Bsa B$ and is therefore admissible in
the optimization defining $\mu_*$ (see \eqref{eq:mu-star}). It thus suffices to show that the Euclidean objective \eqref{eq:mu-star} for this specific $W_{\sigma, \mathsf{B}}$ satisfies
\begin{equation}
    K(W_{\sigma, \mathsf{B}})^2\norm{Q_{\cN,W_{\sigma, \mathsf{B}}}}_\infty
    \leq
    2 \Tr V_B.
\end{equation}

To this end, we first estimate the trace norm domination cost \eqref{eq:KW}. Using Lemma~\ref{lemma:KW-dual},
\begin{align}
    K(W_{\sigma, \mathsf{B}})^2
    &=
    \sup_{\substack{H=H^\dagger\\\norm{H}_\infty\leq1}}
    \Tr\!\left[H W_{\sigma, \mathsf{B}}^{-1}(H)\right] 
    =
    \sup_{\norm{H}_\infty\leq1}\Tr(\sigma H^2)
    \leq1,
\end{align}
with equality attained at $H=\iden_B$. Hence, $K(W_{\sigma, \mathsf{B}})=1$.

It remains to bound the corresponding singular operator \eqref{eq:Q_N,W}. We first diagonalize $\sigma=\sum_i s_i\ketbra{i}$, where $s_i>0$. Since $W_{\sigma, \mathsf{B}}(\ketbra{i}{j})
    = 2\ketbra{i}{j}/(s_i+s_j)$,
applying the expansion \eqref{eq:Q-W-basis-expansion} from Lemma~\ref{lemma:Q-general-W-properties} to the Hilbert-Schmidt
orthonormal basis
\begin{equation}
    \ketbra{i}{i},
    \qquad
    \frac{\ketbra{i}{j}+\ketbra{j}{i}}{\sqrt2},
    \qquad
    \frac{\mathrm{i}(\ketbra{i}{j}-\ketbra{j}{i})}{\sqrt2},
    \qquad i<j,
\end{equation}
gives
\begin{equation}\label{eq:Q-sym-matrix-units}
    Q_{\cN,W_{\sigma, \mathsf{B}}}
    =
    2\sum_{i,j}
    \frac{
        \cN^*(\ketbra{i}{j})\cN^*(\ketbra{j}{i})
    }{s_i+s_j}.
\end{equation}
Write the Choi operator \eqref{eq:choi-def} in blocks with respect to the eigenbasis
$\{\ket{i}\}_i$ of $\sigma$:
\begin{equation}\label{eq:choi-output-blocks}
    J_{\cN}
     =
    \sum_{i,j}X_{ij}\otimes \ketbra{i}{j},
    \qquad
    X_{ij}\in\B{A}.
\end{equation}
Since $J_{\cN}$ is self-adjoint, its blocks satisfy
$X_{ji}=X_{ij}^{\dagger}$. Moreover, $\cN^*(\ketbra{i}{j})=X_{ji}^{\mathsf T},$ where the transpose is taken in the input basis \eqref{eq:choi-def} used to define
$J_{\cN}$. Indeed, the matrix elements of both sides satisfy
$\bra{b}\cN^*(\ketbra{i}{j})\ket{a}
=\bra{j}\cN(\ket{a}\!\bra{b})\ket{i}$.
Substituting this in
\eqref{eq:Q-sym-matrix-units} gives
\begin{equation}\label{eq:Q-sym-choi-blocks}
    Q_{\cN,W_{\sigma, \mathsf{B}}}^{\mathsf T}
    =
    2\sum_{i,j}
    \frac{X_{ij}X_{ji}}{s_i+s_j}
    \leq
    2\sum_{i,j}
    \frac{X_{ij}X_{ji}}{s_j}
    \leq
    2c\,\iden_A,
\end{equation}
where we used $X_{ij}X_{ji}=X_{ij}X_{ij}^{\dagger}\geq 0$, $1/(s_i+s_j)\leq 1/s_j$ and the final inequality is precisely \eqref{eq:Jsigma J<=J} when expanded in the basis \eqref{eq:choi-output-blocks}. Since $K(W_{\sigma, \mathsf{B}})= 1$, we get
\begin{equation}
    \mu_*(\cN)
    \leq
    K(W_{\sigma, \mathsf{B}})^2\norm{Q_{\cN,W_{\sigma, \mathsf{B}}}}_\infty = \norm{Q^{\mathsf T}_{\cN,W_{\sigma, \mathsf{B}}}}_\infty
    \leq 2c =
    2 \Tr V_B.
\end{equation}
\end{proof}

Note that the Euclidean structure \eqref{eq:W-bures} used above is exactly the \emph{Bures} inner product, which is well known in the theory of quantum $\chi^2$-divergences, quantum Fisher information, and monotone Riemannian metrics
\cite{BraunsteinCaves1994,Petz1996,LesniewskiRuskai1999,
PetzGhinea2011,HiaiPetz2012}. There is a precise sense in which the Bures geometry is
optimal in the proof of Lemma~\ref{lemma:mu-star-Imax}. Recall that a (normalized) \emph{monotone metric} on a matrix space is a family
of complex inner products \cite{Petz1996, PetzGhinea2011}
\begin{equation}
    \gamma_\sigma:
    \B{B}\times\B{B}\to\C,
    \qquad
    \sigma\in\cD_+(B),
\end{equation}
which depends continuously on $\sigma$, is real on pairs of
self-adjoint operators, agrees with the classical Fisher metric on
commuting self-adjoint operators,
\begin{equation}
    [X,\sigma]=0,
    \quad
    X=X^\dagger
    \quad\Longrightarrow\quad
    \gamma_\sigma(X,X)
    =
    \Tr(\sigma^{-1}X^2),
\end{equation}
and contracts under quantum channels:
\begin{equation}
    \gamma_{\Phi(\sigma)}
    \bigl(\Phi(X),\Phi(X)\bigr)
    \leq
    \gamma_\sigma(X,X).
\end{equation}
Such monotone metrics are in one-to-one correspondence with operator-monotone functions $f: (0,\infty)\to (0,\infty)$ satisfying $f(1)=1$ and $ 
    f(t)=t f(t^{-1})$ \cite{Petz1996} \cite[Theorem 2.1]{PetzGhinea2011}:
\begin{equation}\label{eq:riemann-metric}
    \gamma_{\sigma,f}(X,Y)
    =
    \Tr\!\left[
        X^{\dagger}\mathsf J_{\sigma,f}^{-1}(Y)
    \right],
    \qquad
    \mathsf J_{\sigma,f}
    =
    R_\sigma^{1/2}
    f(L_\sigma R_\sigma^{-1})
    R_\sigma^{1/2},
\end{equation}
where $L_{\sigma}$ and $R_{\sigma}$ are the commuting positive left and right multiplication operators on the complex space $\B{B}$ defined in \eqref{eq:W-bures}, so that
$f(L_\sigma R_\sigma^{-1})$ is defined by the usual spectral functional
calculus. The symmetry condition $f(t)=t f(t^{-1})$ ensures that
$\mathsf J_{\sigma,f}$ and $W_{\sigma,f}
    :=
    \mathsf J_{\sigma,f}^{-1}$ preserve adjoints. Thus, the
restriction of $W_{\sigma,f}$ to $\Bsa B$ defines a
Euclidean structure in the sense of \eqref{eq:W-Euclidean}. The Bures metric corresponds to the arithmetic-mean function
\begin{equation}
    f_{\mathrm B}(t):=\frac{1+t}{2},
    \qquad
    \mathsf J_{\sigma,\mathrm B}
    =
    \frac{L_\sigma+R_\sigma}{2},
\end{equation}
and $W_{\sigma, B}:= \mathsf J_{\sigma, B}^{-1}$ coincides with the Euclidean structure \eqref{eq:W-bures} used in the proof of Lemma~\ref{lemma:mu-star-Imax}. Since every operator-monotone function $f$ as above satisfies
$f(t)\leq f_{\mathrm B}(t)$ \cite[Section~1.2]{PetzGhinea2011}, 
\begin{equation}
    \mathsf J_{\sigma,f}
    \leq
    \mathsf J_{\sigma,\mathrm B},
    \qquad
    W_{\sigma,\mathrm B}
    \leq
    W_{\sigma,f}.
\end{equation}
Moreover, every such normalized monotone metric has the same trace-norm
domination cost. Indeed, by Lemma~\ref{lemma:KW-dual}, we can express
\begin{align}
    K(W_{\sigma,f})^2
    &=
    \sup_{\substack{H=H^\dagger\\\norm{H}_\infty\leq1}}
    \Tr\!\left[H\mathsf J_{\sigma,f}(H)\right] \\
    &\leq
    \sup_{\norm{H}_\infty\leq1}
    \Tr\!\left[H\mathsf J_{\sigma,\mathrm B}(H)\right]
    =
    \sup_{\norm{H}_\infty\leq1}\Tr(\sigma H^2)
    \leq1,
\end{align}
with equality attained at $H=\iden_B$. Since the covariance map
$W\mapsto Q_{\cN,W}$ is positivity-preserving according to Lemma~\ref{lemma:Q-general-W-properties}, it follows that $Q_{\cN,W_{\sigma,\mathrm B}}
    \leq
    Q_{\cN,W_{\sigma,f}}.$ Consequently,
\begin{equation}
    K(W_{\sigma,\mathrm B})^2
    \norm{Q_{\cN,W_{\sigma,\mathrm B}}}_\infty
    \leq
    K(W_{\sigma,f})^2
    \norm{Q_{\cN,W_{\sigma,f}}}_\infty.
\end{equation}
Thus, for every fixed channel $\cN$ and reference state $\sigma\in\cD_+(B)$, the Euclidean structure induced by the Bures metric minimizes the converse objective in \eqref{eq:mu-star} among all Euclidean structures $W_{\sigma,f}=\mathsf J_{\sigma,f}^{-1}$ arising from monotone metrics via \eqref{eq:riemann-metric}. Note that the unrestricted optimization \eqref{eq:mu-star}, however, allows arbitrary Euclidean structures that need not arise in this way and may, in principle, yield a smaller objective value.

We are now ready to establish the universal $C_{\ID} \leq C_E$ converse bound. The additional technical step on top of Lemma~\ref{lemma:mu-star-Imax} is appropriate smoothing followed by the application of the asymptotic equipartition property \eqref{eq:channel-Imax-AEP}. 

\begin{theorem}[Entanglement-assisted capacity converse on identification]
\label{theorem:CID-leq-CE} \hspace{2pt} \\
Let $\cN:\B{A}\to\B{B}$ be a quantum channel. Then, for every
$\lambda_1,\lambda_2>0$ with $\lambda_1+\lambda_2<1$,
\begin{equation}\label{eq:CID-leq-CE-strong-converse}
    C_{\ID}(\cN)\leq  \limsup_{n\to\infty}
    \frac1n
    \log\log N_{(n,\lambda_1,\lambda_2)}(\cN)
    \leq
    C_E(\cN).
\end{equation}
\end{theorem}

\begin{proof}
Set $\Delta:=1-\lambda_1-\lambda_2>0,$
and fix $0<\eps<\Delta/4$. For every $n\in\N$, choose a
channel $\widetilde{\cN}_n:
    \B{A^{\otimes n}}
    \to
    \B{B^{\otimes n}}$
such that
\begin{equation}\label{eq:smooth-channel-choice}
    \frac12
    \norm{\widetilde{\cN}_n-\cN^{\otimes n}}_\diamond
    \leq
    \eps,
    \qquad
    I_{\max}(\widetilde{\cN}_n)
    \leq
    I_{\max}^{\eps}(\cN^{\otimes n})+1.
\end{equation}
Now, consider a $(n,N,\lambda_1,\lambda_2)$ identification code
$(\rho_i,D_i)_{i\in[N]}$ for $\cN$, and define
\begin{equation}
    \omega_i:=\cN^{\otimes n}(\rho_i),
    \qquad
    \widetilde{\omega}_i:=\widetilde{\cN}_n(\rho_i).
\end{equation}
The diamond-norm bound in
\eqref{eq:smooth-channel-choice} implies $d_{\Tr}(\omega_i,\widetilde{\omega}_i)
    \leq
    \eps$ for every $i$. Consequently, for every effect
$0\leq D\leq\iden_{B^{\otimes n}}$,
   $ \left|
        \Tr\bigl((\widetilde{\omega}_i-\omega_i)D\bigr)
    \right|
    \leq
    \eps.$ It follows that
\begin{align}
    \Tr(\widetilde{\omega}_iD_i)
    &\geq
    \Tr(\omega_iD_i)-\eps
    \geq
    1-\lambda_1-\eps, \label{eq:smoothed-ID-first-kind}\\
    \Tr(\widetilde{\omega}_iD_j)
    &\leq
    \Tr(\omega_iD_j)+\eps
    \leq
    \lambda_2+\eps,
    \qquad i\neq j. \label{eq:smoothed-ID-second-kind}
\end{align}
Thus, the same family $(\rho_i,D_i)_{i\in[N]}$ is a
$(1,N,\lambda_1+\eps,\lambda_2+\eps)$
identification code for $\widetilde{\cN}_n$. Equivalently, the states
$\{\widetilde{\omega}_i\}_{i\in[N]}$ form a
$\zeta$-packing of the image (see Lemma~\ref{lemma:CID-separation})
\begin{equation}
    \widetilde I_n
    :=
    \widetilde{\cN}_n\bigl(\cD(A^{\otimes n})\bigr)
\end{equation}
in trace distance for every $0<\zeta<
    1-(\lambda_1+\eps)-(\lambda_2+\eps)
    =
    \Delta-2\eps.$ Fix $\zeta:=(\Delta-2\eps)/2>0.$
Applying 
Theorem~\ref{theorem:fully-optimized-one-shot-converse} to $\widetilde{\cN}_n$ at blocklength one shows that
\begin{align}
    \log N
    \leq
    \log
    N_{(1,\lambda_1+\eps,\lambda_2+\eps)}
    (\widetilde{\cN}_n)
    \leq
    \frac{2Cn\log d_A}{\zeta^2}\,
    \mu_*(\widetilde{\cN}_n).
    \label{eq:smoothed-one-shot-mu-star}
\end{align}
Since the above estimate holds for every
$(n,N,\lambda_1,\lambda_2)$ identification code for $\cN$, we obtain
\begin{equation}
    \log N_{(n,\lambda_1,\lambda_2)}(\cN)
    \leq
    \frac{2Cn\log d_A}{\zeta^2}\,
    \mu_*(\widetilde{\cN}_n) \leq \frac{8Cn\log d_A}{\zeta^2}\,
    2^{I_{\max}^{\eps}(\cN^{\otimes n})}.
\end{equation}
where the final bound follows from Lemma~\ref{lemma:mu-star-Imax} and \eqref{eq:smooth-channel-choice}.
Thus,
\begin{align}
    \frac1n
    \log\log N_{(n,\lambda_1,\lambda_2)}(\cN)
    \leq
    \frac1n
    I_{\max}^{\eps}(\cN^{\otimes n})
    +
    \frac1n
    \log\left(
        \frac{8Cn\log d_A}{\zeta^2}
    \right),
\end{align}
where the second term vanishes as
$n\to\infty$. Therefore,
\begin{equation}\label{eq:ID-rate-via-smooth-Imax}
    \limsup_{n\to\infty}
    \frac1n
    \log\log N_{(n,\lambda_1,\lambda_2)}(\cN)
    \leq
    \limsup_{n\to\infty}
    \frac1n
    I_{\max}^{\eps}(\cN^{\otimes n}).
\end{equation}
The estimate holds for every $0<\eps<\Delta/4$. Letting
$\eps\to0$ and applying the channel max-information asymptotic
equipartition property
\eqref{eq:channel-Imax-AEP}, we obtain
\begin{equation}
    \limsup_{n\to\infty}
    \frac1n
    \log\log N_{(n,\lambda_1,\lambda_2)}(\cN)
    \leq
    C_E(\cN).
\end{equation}
\end{proof}

The following is a simple corollary of \eqref{eq:CID>=QID>=C_E} and Theorem~\ref{theorem:CID-leq-CE}.

\begin{corollary}\label{cor:low-noise}
Let $\cN:\B{A}\to \B{B}$ be a low noise channel (see Definition~\ref{def:low-noise}). Then,
\begin{equation}
    C_{\ID}(\cN)=Q_{\ID, v}(\cN)
    =Q_{\ID}(\cN)
    =C_E(\cN).
\end{equation}
\end{corollary}

\subsection{Strictness and superadditivity}

We now show that the converse bound from Theorem~\ref{theorem:CID-leq-CE} can be strict $C_{\ID}<C_E$, thus proving that $C_{\ID}\geq C_E$ \emph{cannot} universally hold. We do this by analyzing the depolarizing $\cD_p:\B{\C^d}\to \B{\C^d}$ and transpose-depolarizing $\cD_q^\top : \B{\C^d}\to \B{\C^d}$ channels:
\begin{align}
    \cD_p(X) &:= (1-p)X + p \Tr(X) \frac{\iden}{d}, \label{eq:depol}  \\
    \cD^\top_q(X) &:= (1-q)X^\top + q \Tr(X) \frac{\iden}{d},  \label{eq:t-depol}
\end{align}
in the respective parameter ranges where the maps are completely positive and trace-preserving: $0\leq p\leq \frac{d^2}{d^2-1}$ and $\frac{d}{d+1}\leq q\leq \frac{d}{d-1}$.

\begin{proposition}\label{prop:main}
Let $\cD_q^\top : \B{\C^d}\to \B{\C^d}$ be the transpose-depolarizing channel. For $q=\frac{d}{d+1}$ and $d\ge 2$, the following holds true:
    \begin{align}\label{eq:CID<=hd+2logd/d+1}
        C_{\ID}\left(\cD_{\frac{d}{d+1}}^{\mathsf T}\right) &\leq \log d-\frac{d-1}{d}\log(d+1) \\
        & < 1 + \log \frac{d}{d+1} = C_E \left(\cD_{\frac{d}{d+1}}^\top \right).
    \end{align}
\end{proposition}
\begin{proof}
Since the conditions \eqref{eq:lambda_1},\eqref{eq:lambda_2} for the existence of an $n$-shot ID code for any channel $\cN$ only depend on the image $I_n(\cN)$ in \eqref{eq:N-image}, it follows that 
\begin{equation}
    C_{\ID}\left(\cD_{\frac{d}{d+1}}^{\mathsf T}\right) = C_{\ID}\left(\cD_{\frac{d}{d+1}}\right) \leq C_E \left(\cD_{\frac{d}{d+1}}\right).
\end{equation}
Indeed, $I_n\left(\cD_{\frac{d}{d+1}}^{\mathsf T}\right)=I_n\left(\cD_{\frac{d}{d+1}}\right)$ for all $n$, since the two channels are related via transposition:
\begin{equation}
    (\cD_p^{\mathsf T})^{\otimes n} = {\mathsf T}^{\otimes n} \circ \cD_p^{\otimes n},
\end{equation}
which maps $\cD((\C^d)^{\otimes n})$ bijectively onto itself. The upper bound in terms of $C_E$ follows from Theorem~\ref{theorem:CID-leq-CE}. It remains to compute $C_E$. Using unitary covariance 
\begin{align}
\cD_p(UXU^\dagger)=
   U\,\cD_p(X) U^\dagger, \qquad
    \cD_q^{\mathsf T}(UXU^\dagger)=
    \overline U\,\cD_q^{\mathsf T}(X)\,\overline U^\dagger,
\end{align}
it is easy to show that the optimizing state in the $C_E$ formula \eqref{eq:C_E=I} is the maximally entangled state. More precisely, let
\begin{align}
\ket{\Phi}_{RA} &=\frac1{\sqrt d}\sum_{i\in [d]} \ket{i}_R\ket{i}_A, \\
    \sigma_{RB} &:= (\id_R\otimes \cD_{\frac{d}{d+1}})(\Phi_{RA}) =
    \frac1d\Phi_{RB}
    +
    \frac{1}{d(d+1)}
    \bigl(\iden_{RB}-\Phi_{RB}\bigr). \\
    \omega_{RB}
    &:=
    (\id_R\otimes \cD_{\frac{d}{d+1}}^{\mathsf T})(\Phi_{RA})
    =
    \frac{2P_{\rm sym}}{d(d+1)},
\end{align}
where $P_{\rm sym}$ is the projector onto the symmetric subspace in $R\otimes B$, where $R \cong A\cong B \cong \C^d$. A straightforward computation then reveals
\begin{align}
    C_E(\cD_{\frac{d}{d+1}})=
    I(R;B)_\sigma &= \log d-\frac{d-1}{d}\log(d+1)  \\
    &< \log \frac{2d}{d+1} =
    I(R;B)_\omega = C_E(\cD_{\frac{d}{d+1}}^{\mathsf T}).
\end{align}
\end{proof}

The non-additivity of $C_{\ID}$ follows as a simple consequence of Proposition~\ref{prop:main}.

\begin{corollary}[Strict superadditivity of \(C_{\ID}\)]
\label{cor:CID-nonadditive}
For integers $d\geq 2$ and \(m>(d+1)/2\),
\begin{equation}
    C_{\ID}(\cD_{\frac{d}{d+1}}^\top\otimes \id_m)
    >
    C_{\ID}(\cD_{\frac{d}{d+1}}^\top)+C_{\ID}(\id_m).
\end{equation}
\end{corollary}

\begin{proof}
Let $A\cong R\cong B \cong \C^d$ and $A'\cong R'\cong B' \cong \C^m$. The product maximally entangled input $\Phi_{RA} \otimes \Phi_{R'A'}$ for
\(\cD_{d/(d+1)}^{\mathsf T}\otimes\id_m\) has output
\begin{equation}
    \omega_{RB} \otimes \Phi_{R'B'} = (\id_R \otimes \cD_{\frac{d}{d+1}}^{\mathsf T})(\Phi_{RA}) \otimes \Phi_{R'B'}
\end{equation}
with coherent information
\begin{equation}
    I(RR'\rangle BB')_{\omega \otimes \Phi}
    =
    \log\frac{2}{d+1}+\log m
    =
    \log\frac{2m}{d+1},    
\end{equation}
which is positive precisely when \(m>(d+1)/2\). Moreover, by additivity of \(C_E\), the same product input is optimal in \eqref{eq:C_E=I} for
the product channel. Thus, for $m>(d+1)/2$, the product channel $\cD^{\mathsf T}_{d/(d+1)}\otimes \id_m$ is low-noise in the sense of Definition~\ref{def:low-noise}. Hence, by using \eqref{eq:CID>=QID>=C_E}, additivity of $C_E$, Proposition~\ref{prop:main}, and Eq.~\eqref{eq:C(id)>Csim(id)}, we get the desired result:
\begin{align}
    C_{\ID}(\cD_{\frac{d}{d+1}}^\top\otimes\id_m)
    &\ge
    C_E(\cD_{\frac{d}{d+1}}^\top\otimes\id_m) \\
    &=
    C_E(\cD_{\frac{d}{d+1}}^\top)+C_E(\id_m) \\
    &>
    C_{\ID}(\cD_{\frac{d}{d+1}}^\top)+2\log m \\
    &=
    C_{\ID}(\cD_{\frac{d}{d+1}}^\top)+C_{\ID}(\id_m).
\end{align}
\end{proof}

\subsection{Comparison with other converse bounds}

The $C_{\ID}\leq C_E$ converse bound is closely related to, but should be
distinguished from, the Euclidean converse bounds of
\cite{singh2026gaussianmeanwidthstrong}. For a full-rank state
$\sigma\in\cD_+(B)$, consider the weighted Euclidean structure $W_{\sigma}:\Bsa{B}\to \Bsa{B}$:
\begin{equation}\label{eq:state-weighted-Euclidean-structure}
    W_\sigma(Y)
    :=
    \sigma^{-1/2}Y\sigma^{-1/2},
\end{equation}
with associated inner product and norm
\begin{equation}
    \langle Y,Z\rangle_\sigma
    :=
    \Tr\bigl(YW_\sigma(Z)\bigr), \qquad \norm{Y}_{\sigma} := \sqrt{\langle Y, Y \rangle_{\sigma} }.
\end{equation}
Use the product structure
$W_{\sigma^{\otimes n}}=W_\sigma^{\otimes n}$ on $\Bsa{B^{\otimes n}}$. Let $\cN:\B{A}\to \B{B}$ be a channel. We denote $Q_{\cN,\sigma}:=
Q_{\cN,W_\sigma}.$ (see Eq.~\eqref{eq:Q_N,W}). Then, H\"older's inequality shows that
$\norm{Y}_1\leq\norm{Y}_{W_\sigma}$, and equality holds for $Y=\sigma$. Hence, $K(W_\sigma)=1$ \cite[Lemma 2.1]{singh2026gaussianmeanwidthstrong}. Moreover, the corresponding singular
operator tensorizes \cite[Lemma 3.3]{singh2026gaussianmeanwidthstrong}:
\begin{equation}
    Q_{\cN^{\otimes n},\sigma^{\otimes n}}
    =
    Q_{\cN,\sigma}^{\otimes n}.
\end{equation}
Choosing this particular Euclidean structure in
Theorem~\ref{theorem:fully-optimized-one-shot-converse}, taking a second
logarithm, dividing by $n$, taking the limit $n\to \infty$, and then optimizing over $\sigma\in \cD_+(B)$ gives the state-weighted single-letter converse bound (see \cite[Theorem 3.8]{singh2026gaussianmeanwidthstrong}):
\begin{equation}\label{eq:state-weighted-Gaussian-converse}
    \limsup_{n\to\infty}
    \frac1n
    \log\log N_{(n,\lambda_1,\lambda_2)}(\cN)
    \leq
    \inf_{\sigma\in\cD_+(B)}
    \log\norm{Q_{\cN,\sigma}}_\infty.
\end{equation}
The $C_E$ converse improves upon this bound, as the following proposition illustrates. 

\begin{proposition}
For every quantum
channel $\cN:\B{A}\to \B{B}$,
\begin{equation}\label{eq:CE-below-state-weighted-Gaussian}
    C_E(\cN)
    \leq
    \inf_{\sigma\in\cD_+(B)}
    \log\norm{Q_{\cN,\sigma}}_\infty.
\end{equation}
\end{proposition}

\begin{proof}
Let $\rho_{RA}$ be an input state, $\omega_{RB}:=(\id_R\otimes\cN)(\rho_{RA})$, so that $\rho_R=\omega_R$. Fix a full-rank state $\sigma\in \cD_+(B)$. Let $\{G_a\}_a$ be a $\sigma$-orthonormal basis in $(\Bsa{B}, \langle \cdot , \cdot \rangle_{\sigma})$. By definition \eqref{eq:Q_N,W} (see also Lemma~\ref{lemma:Q-general-W-properties}), we can write the singular operator as
\begin{equation}
    Q_{\cN,\sigma} = \sum_a \bigl( \cN^{*,\sigma}  (G_a) \bigr)^2 = \sum_a M_a^2,
\end{equation}
where $\cN^{*,\sigma}:= \cN^{*, W_{\sigma}}=\cN^*\circ W_{\sigma}$ is the $W_{\sigma}$-adjoint \eqref{eq:N-Wadjoint}, and $
M_a:=\cN^{*,\sigma}\!\left(G_a\right)\in \Bsa{A}$.
We expand $\omega_{RB}$ in the basis $\{G_a \}_a$:
\begin{equation}
\omega_{RB}=\sum_a X_a\otimes G_a,
\end{equation}
so that weighted orthonormality gives $X_a
=
\Tr_A\!\left[(\iden_R\otimes M_a)\rho_{RA}\right] =: \cT_\rho(M_a)$.
Consequently,
\begin{align}
\widetilde{Q}_2(\omega_{RB}\|\rho_R\otimes\sigma_B)
&:= \Tr \left[ \left( (\rho_R^{-\frac{1}{4}} \otimes \sigma_B^{-\frac{1}{4}}) \omega_{RB} (\rho_R^{-\frac{1}{4}} \otimes \sigma_B^{-\frac{1}{4}} )\right)^2 \right] \\ 
&= \sum_a
\Tr\!\left(
X_a\rho_R^{-1/2}X_a\rho_R^{-1/2}
\right),
\label{eq:Q2-expansion}
\end{align}
where inverse powers are taken on $\supp\rho_R$. Let $H_a:=\rho_R^{-1/2}X_a\rho_R^{-1/2}.$ The map
\begin{equation}
M\longmapsto
\rho_R^{-1/2}\cT_\rho(M)\rho_R^{-1/2}
\end{equation}
is positive and unital on $\supp\rho_R$. Hence, by Kadison's
inequality \cite[Theorem 2.3.2]{bhatia2015positive}, we can write $H_a^2
\leq
\rho_R^{-1/2}\cT_\rho(M_a^2)\rho_R^{-1/2}.$ Moreover,
\begin{align}
\Tr\!\left(
X_a\rho_R^{-1/2}X_a\rho_R^{-1/2}
\right)
&=
\Tr\!\left(
\rho_R^{1/2}H_a\rho_R^{1/2}H_a
\right) \\
&\leq
\Tr(\rho_R H_a^2) \leq
\Tr\!\left(\cT_\rho(M_a^2)\right) =
\Tr(\rho_A M_a^2),
\end{align}
where the first inequality follows from
\begin{equation}
\Tr(\rho_RH_a^2)
-
\Tr\!\left(\rho_R^{1/2}H_a\rho_R^{1/2}H_a\right)
=
\frac12
\norm{[\rho_R^{1/2},H_a]}_2^2
\geq 0.
\end{equation}
Together with \eqref{eq:Q2-expansion}, this yields
\begin{align}
\widetilde Q_2(\omega_{RB}\|\rho_R\otimes\sigma_B) \leq
\sum_a\Tr(\rho_A M_a^2) =
\Tr\!\left(\rho_AQ_{\cN,\sigma}\right) \leq
\norm{Q_{\cN,\sigma}}_\infty.
\end{align}
Therefore, using the ordering between the quantum relative entropy and sandwiched-R\'enyi $2$-divergence $D\leq \widetilde{D}_2$ \cite{MullerLennert2013sandwich, Beigi2013}, we obtain
\begin{align}
I(R:B)_\omega
\leq
D(\omega_{RB}\|\rho_R\otimes\sigma_B) 
&\leq
\widetilde D_2(\omega_{RB}\|\rho_R\otimes\sigma_B) \\ 
&=:
\log\widetilde Q_2(\omega_{RB}\|\rho_R\otimes\sigma_B) 
\leq
\log\norm{Q_{\cN,\sigma}}_\infty.
\end{align}
Since the input $\rho$ and the weighing state $\sigma$ were arbitrary, we get the desired claim.
\end{proof}

Comparison with the fully optimized Euclidean bound \cite[Theorem 5.2]{singh2026gaussianmeanwidthstrong}
\begin{equation}
 \limsup_{n\to\infty}
    \frac1n
    \log\log N_{(n,\lambda_1,\lambda_2)}(\cN) \leq   \limsup_{n\to\infty}
    \frac1n\log\mu_*(\cN^{\otimes n})
\end{equation}
is unclear, since we do not currently have a universal ordering between this unsmoothed regularized quantity and $C_E(\cN)$. In the proof of Theorem~\ref{theorem:CID-leq-CE}, the Euclidean converse is instead applied to a channel that is diamond-norm close to $\cN^{\otimes n}$. The estimate
$\mu_*(\cM)\leq2^{I_{\max}(\cM)+1}$ and the smooth max-information AEP then turn this \emph{smoothed} Euclidean converse into the single-letter bound $C_E(\cN)$. The Euclidean quantities nevertheless retain finer finite-block and channel-image geometric information, and might yield sharper channel-specific bounds than the entanglement-assisted capacity.

The $C_{\ID}\leq C_E$ also improves upon the quantum capacity converse bound from \cite{Atif2024CIDstrongconverse}.

\begin{proposition}
    For every quantum
channel $\cN:\B{A}\to \B{B}$,
\begin{equation}
    C_E(\cN)
    \leq \log d_A + Q^{\dagger}(\cN),
\end{equation}
where $Q^{\dagger}$ is the strong converse quantum capacity \cite{Atif2024CIDstrongconverse}.
\end{proposition}
\begin{proof}
For an input pure state $|\phi\rangle_{RA}$, set $\rho_{RB}=(\id_R\otimes\cN_{A\to B})(\phi_{RA})$. Then, 
\begin{equation}
    I(R:B)_{\rho} = S(R)_{\rho} + I(R\rangle B)_{\rho} \leq \log d_A + Q^{(1)}(\cN), 
\end{equation}
where 
\begin{equation}
    Q^{(1)}(\cN)= \sup_{\phi}
I(R\rangle B)_\rho 
\end{equation}
is the channel's coherent information, which regularizes to the quantum capacity: 
\begin{equation}
    Q(\cN) = \sup_{n\in \N} \frac{1}{n} Q^{(1)}(\cN^{\otimes n}) \leq Q^{\dagger}(\cN).
\end{equation}
\end{proof}

We compare the known converse and achievability bounds for the qubit depolarizing channel $\cD_p: \M{2}\to \M{2}$ in Figure~\ref{fig:qubit-depolarizing}, illustrating how the $C_E$ converse improves upon all the previously known converse bounds in the full parameter range $0\leq p\leq 1$. For detailed computation of these bounds, we refer the reader to \cite[Section 4.1.1]{singh2026gaussianmeanwidthstrong}.

\begin{figure}[ht]
    \centering
    \includegraphics[width=0.69\linewidth]{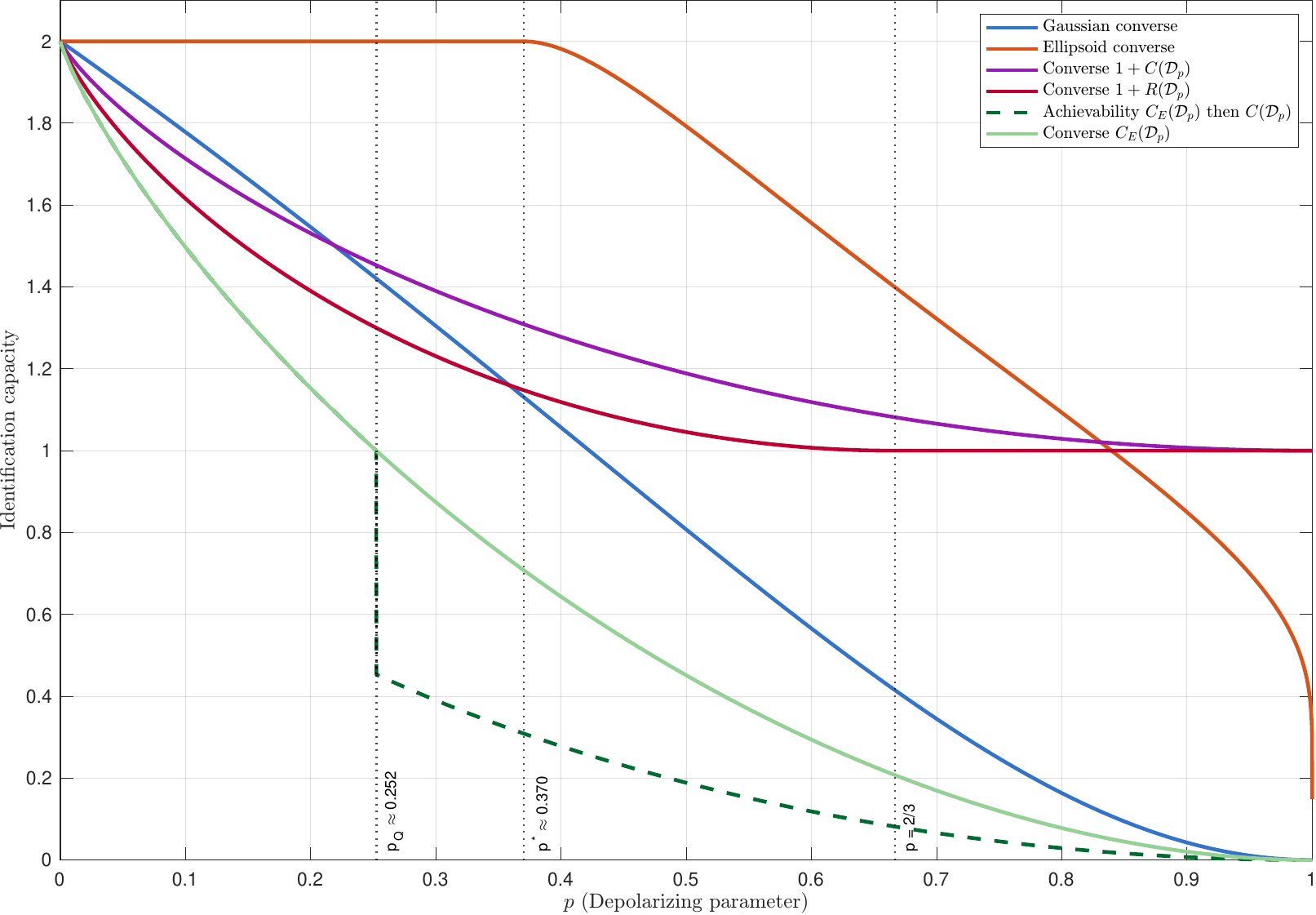}
    \caption{Strong converse and achievability bounds on the classical identification capacity of the qubit depolarizing channel $\cD_p$ (see Eq.~\eqref{eq:depol}) for $0\leq p \leq 1$. The light green curve shows the $C_E$ converse bounds from Theorem~\ref{theorem:CID-leq-CE}. The blue curve shows the state-weighted Gaussian converse bound from \cite[Theorem 3.8]{singh2026gaussianmeanwidthstrong} (see \eqref{eq:state-weighted-Gaussian-converse}). The orange curve shows the Ellipsoid converse bound from \cite{ye2026strongconverseboundsclassical}. The purple and red curves show the classical capacity \cite{ye2026strongconverseboundsclassical} and quantum capacity \cite{Atif2024CIDstrongconverse} converse bounds, respectively. The dashed curve shows the achievability bound from \cite{Hayden2012QID-achievability, Winter2013survey-ID} (see \eqref{eq:CID>=C}, \eqref{eq:CID>=QID>=C_E}). Note that the achievability and the $C_E$ converse coincide for $p\lesssim 0.252$ in accordance with Corollary~\ref{cor:low-noise}.}
    \label{fig:qubit-depolarizing}
\end{figure} 

\section{Discussion}
\label{sec:discussion}

In this paper, we have established a universal strong converse bound on classical identification: $C_{\ID}\leq C_E$. This improves upon the two main converse bounds from the literature, namely the capacity-based converse from \cite{Atif2024CIDstrongconverse} and the recently obtained state-weighted Euclidean converse from \cite{singh2026gaussianmeanwidthstrong}. Its ordering relative to the fully optimized unsmoothed Euclidean converse from \cite[Theorem 5.2]{singh2026gaussianmeanwidthstrong} is presently
unknown. For low-noise channels (Definition~\ref{def:low-noise}), our bound is achievable \cite{Hayden2012QID-achievability, Winter2013survey-ID}, implying $C_{\ID}=C_E$ for such channels. However, for highly noisy channels, we have shown that the bound can be strict $C_{\ID}< C_E$.

The image invariance property of identification that we used in the transpose-depolarizing example (Proposition~\ref{prop:main}) leads to a more intrinsic formulation of the bounds studied in this paper. We formally note this invariance property below.

\begin{lemma}[Image-set invariance of \(C_{\ID}\)]
\label{lemma:CID-image-invariance}
If two channels \(\cN,\cM:\B{A}\to\B{B}\) satisfy
\(I_n(\cN)=I_n(\cM)\) for some \(n\), then, for every
\(\lambda_1,\lambda_2\in[0,1)\),
\begin{equation}
N_{(n,\lambda_1,\lambda_2)}(\cN)
=
N_{(n,\lambda_1,\lambda_2)}(\cM).
\end{equation}
Consequently, if the equality of image sets holds for every \(n\), then
\(C_{\ID}(\cN)=C_{\ID}(\cM)\).
\end{lemma}

\begin{proof}
The claim follows since the defining conditions \eqref{eq:lambda_1},\eqref{eq:lambda_2} of an $(n,N,\lambda_1,\lambda_2)$ ID code for a channel $\cN$ only depends on the channel image $I_n(\cN)$ \eqref{eq:N-image}. 
\end{proof}

For channels $\cN,\cM$ with a common output space, define the tensor-image preorder
\begin{align}
    \cN\preceq_{\mathrm{im}}\cM
    &\quad :\Longleftrightarrow\quad
    I_n(\cN)\subseteq I_n(\cM)
    \quad\text{for every }n\in\N.\label{eq:image-preorder}
\end{align}
The proof of Lemma~\ref{lemma:CID-image-invariance} immediately gives the stronger monotonicity statement
\begin{equation}\label{eq:CID-image-monotonicity}
    \cN\preceq_{\mathrm{im}}\cM
    \quad\Longrightarrow\quad
    N_{(n,\lambda_1,\lambda_2)}(\cN)
    \leq
    N_{(n,\lambda_1,\lambda_2)}(\cM)
\end{equation}
for every $n\in \N$ and $\lambda_1,\lambda_2\in [0,1)$, and hence
$C_{\ID}(\cN)\leq C_{\ID}(\cM)$. This permits the achievability \eqref{eq:CID>=C}, \eqref{eq:CID>=QID} and converse (Theorem~\ref{theorem:CID-leq-CE}) bounds to be optimized over sub- and super-realizations of the asymptotic output geometry. Define
\begin{align}
    \underline{C}_{\mathrm{im}}(\cN)
    &:=
    \sup_{k\geq1}\frac1k
    \sup_{\cM:\,\cM\preceq_{\mathrm{im}}\cN^{\otimes k}}
    \max\bigl\{C(\cM),Q_{\ID,v}(\cM)\bigr\}, \label{eq:Cim-lower} \\
    \overline{C}_{\mathrm{im}}(\cN)
    &:=
    \inf_{k\geq1}\frac1k
    \inf_{\cM:\,\cN^{\otimes k}\preceq_{\mathrm{im}}\cM}
    C_E(\cM). \label{eq:Cim-upper}
\end{align}
In these optimizations, $\cM$ may have an arbitrary input space but has the same output space as the block channel $\cN^{\otimes k}$. These quantities give the following image-optimized bounds.

\begin{theorem}\label{theorem:im-opt-bounds}
    Let $\cN:\B{A}\to \B{B}$ be a quantum channel. Then,  
    \begin{equation}\label{eq:Cim-sandwich}
    \max\bigl\{C(\cN),Q_{\ID,v}(\cN)\bigr\}
    =
    \underline{C}_{\mathrm{im}}(\cN)
    \leq
    C_{\ID}(\cN)
    \leq
    \overline{C}_{\mathrm{im}}(\cN)
    \leq
    C_E(\cN).
\end{equation}
\end{theorem}
\begin{proof}
    Reblocking shows $C_{\ID}(\cN^{\otimes k})=kC_{\ID}(\cN)$. If $\cN^{\otimes k}\preceq_{\mathrm{im}}\cM$, then image monotonicity \eqref{eq:CID-image-monotonicity} and the universal $C_E$ converse (Theorem~\ref{theorem:CID-leq-CE}) give
\begin{equation}
    kC_{\ID}(\cN)
    =C_{\ID}(\cN^{\otimes k})
    \leq C_{\ID}(\cM)
    \leq C_E(\cM),
\end{equation}
which proves the upper bound after optimizing. For the lower bound, notice that since the defining conditions of an $n$-shot classical transmission code \eqref{eq:lambda} and a visible quantum identification code \eqref{eq:QID-eps} only depend on the image $I_n(\cN)$, both $C$ and $Q_{\ID, v}$ are themselves image-monotone:
\begin{equation}
    \cN \preceq_{\im} \cM \implies C(\cN) \leq C(\cM) , \,\, Q_{\ID, v}(\cN) \leq Q_{\ID, v}(\cM).
\end{equation}
Hence, image optimization cannot improve the known achievability bounds \eqref{eq:CID>=C}, \eqref{eq:CID>=QID}: $\underline{C}_{\mathrm{im}}(\cN)= \max \{ C(\cN) , Q_{\ID, v}(\cN)\}$.
\end{proof}

It is unclear how difficult it is to actually compute the image-optimized upper bound \eqref{eq:Cim-upper} in practice. For low-noise channels (Definition~\ref{def:low-noise}), we have already seen that the lower and upper bounds collapse to $C_E$ (Corollary~\ref{cor:low-noise}). Outside this
regime, image optimization can strictly improve the $C_E$ converse, as illustrated by the transpose-depolarizing channel (Proposition~\ref{prop:main}). Understanding when this image-optimized converse is tight is therefore a natural direction for further investigation.

\section*{Acknowledgements}

I thank Andreas Winter and Pau Colomer for several insightful discussions on identification. Motivated by these discussions. I asked ChatGPT 5.6 Sol to use the geometric Euclidean converse functional $\mu_\star$ from \cite[Eq. (5.24)]{singh2026gaussianmeanwidthstrong} to try to establish the converse bound $C_{\ID} \leq C_E$. After a few iterations of guided discussions, the model came up with the right Euclidean structure \eqref{eq:W-bures} which connects $\mu_\star$ to the channel max information $I_{\max}$ \cite{Fang2020channel-max-information} via Lemma~\ref{lemma:mu-star-Imax}, after which smoothing combined with a straightforward application of the asymptotic equipartition property \eqref{eq:channel-Imax-AEP} yields the desired bound. I have verified all AI-assisted material and take full responsibility for this work.

I acknowledge support from the Deutsche Forschungsgemeinschaft (DFG, German Research Foundation) via TRR 352 – Project-ID 470903074. 

\bibliographystyle{plainurl}
\bibliography{references}

\end{document}